\documentclass[12pt]{iopart}

\usepackage{graphicx}
\usepackage{amssymb}
\expandafter\let\csname equation*\endcsname\relax

\expandafter\let\csname endequation*\endcsname\relax
\usepackage[cmex10]{amsmath}
\usepackage{amsthm}
\usepackage{mathtools}
\usepackage{pgfplots, filecontents}
\pgfplotsset{compat=1.12}
\usepackage{epstopdf}
\graphicspath{{./}}
\usepackage{cite}
\usepackage{enumerate}
\usepackage{bm}
\usepackage{subfigure}
\usepackage{bigstrut}
\usepackage{multirow}
\usepackage{hhline}

\newcommand{\ve}{\mathbf}
\newcommand{\m}{\mathbf}

\newtheorem{theorem}{Theorem}[section]
\newtheorem{corollary}{Corollary}[theorem]
\newtheorem{lemma}[theorem]{Lemma}

\newcommand{\IR}{\mathbb{R}}
\DeclareMathOperator{\diag}{diag}

\begin{document}

\title[A Linear State Space Model for Photoacoustic Imaging in an Acoustic Attenuating Media]{A Linear State Space Model for Photoacoustic Imaging in an Acoustic Attenuating Media}

\author{Oliver~Lang$^1$, P\'eter~Kov\'acs$^{1,3}$, Christian~Motz$^1$, Mario~Huemer$^1$, Thomas~Berer$^2$, Peter~Burgholzer$^2$}

\address{$^1$Institute of Signal Processing, Johannes Kepler University, Linz, Austria}
\address{$^2$Research Center for Non Destructive Testing (RECENDT), Linz, Austria}
\address{$^3$Department of Numerical Analysis, E\"otv\"os L. University, Budapest, Hungary}

%Recommended referees: Patrick J. La Riviére (Physics)
%					 						 Mark A. Anastasio (Physics)
%					  					 Luís Deán-Ben (Physics)
%					 						 Cox Bradley T. Treeby (Physics)
%					 						 Leonid Kunyansky (Mathematics)
%					 						 Markus Haltmeier (Mathematics)
					
%\ead{submissions@iop.org}
\vspace{10pt}
\begin{indented}
\item[]July 2018
\end{indented}

\begin{abstract}
In photoacoustic imaging, ultrasound waves generated by a temperature rise after illumination of light absorbing structures are measured on the sample surface. These measurements are then used to reconstruct the optical absorption. We develop a method for reconstructing the absorption inside the sample based on a discrete linear state space reformulation of a partial differential equation that describes the propagation of the ultrasound waves. Fundamental properties of the corresponding state space model such as stability, observability and controllability are also analyzed. By using Stokes' equation, the frequency dependent attenuation of the ultrasound waves is incorporated into our model, therefore the proposed method is of general nature. As a consequence, this approach allows for inhomogeneous probes with arbitrary absorption profiles and it accounts for the decrease in laser intensity due to absorption. Furthermore, it provides a method for optimizing the laser modulation signal such that the accuracy of the estimated absorption profile is maximized. Utilizing the optimized laser modulation signal yields an increase in reconstruction accuracy compared to short laser pulses as well as chirp modulation in many scenarios. 
\end{abstract}

\vspace{2pc}
\noindent{\it Keywords}: Photoacoustic image reconstruction, acoustic attenuation, linear state space model, optimal temporal laser excitation.
%\maketitle

% -------------------------------------------------------------------
% Introduction
% -------------------------------------------------------------------
\section{Introduction}
\label{sec:intro}

Photoacoustic imaging, also called optoacoustic or thermoacoustic imaging, is based on the generation of ultrasound following a temperature rise after illumination of light absorbing structures within a (semi)transparent and turbid material, such as a biological tissue. It provides optical images with specific absorption contrast \cite{Kruger2003, Beard2011, Wang2012}. Therefore, it offers greater specificity than conventional ultrasound imaging with the ability to detect hemoglobin, lipids, water and other light-absorbing chromophores, but with greater penetration depth than purely optical imaging modalities that rely on ballistic photons. In photoacoustic tomography the temporal evolution of the acoustic pressure field is sampled using an array of ultrasound detectors placed on or outside the tissue surface or by moving a single detector across the detection surface. Images of the optical absorption within the tissue are then reconstructed by solving an inverse source problem \cite{Wang2012, Burgholzer2007, Kuchment2008}. 

Usually, for illumination short laser pulses are used, where the generated acoustic pressure just after the pulse is proportional to the absorbed optical energy density \cite{Burgholzer2007}. Nevertheless, for the last decade also intensity-modulated continuous-wave lasers have been used for optical excitation of the ultrasound \cite{Maslov2008, Langer:16, Mienkina5507660, Su:11, Telenkov:11}. Temporal modulation of the excitation laser can also be applied for maximizing the accuracy and spatial resolution of the reconstructed image \cite{burg2018}. Various excitation schemes for frequency-domain photoacoustic tomography have been used in experiments and analyzed theoretically \cite{Baddour2008, Mohajerani2014}. In comparison to time-domain excitation, the generation of acoustic pressure waves is usually much less effective for frequency-domain excitation \cite{PMID:25302158}. To compensate for this, pulse compression techniques using frequency chirps and matched filtering are frequently employed \cite{Telenkov2009} and were compared to pulsed excitation \cite{Telenkov2010, Lashkari2011, Petschke_2010}. If technical limitations from acoustic detectors, called transducers, or amplifiers are neglected, short pulses always give better results using the same light energy or limiting the light fluence by the American National Standards Institute laser safety guidelines for skin \cite{Petschke_2010}. This is true if acoustic attenuation in the sample tissue can be neglected. 

In this work, the influence of acoustic attenuation for choosing the optimal temporal laser excitation is investigated. At depths larger than the range of the ballistic photons, i.e. more than a few hundreds of microns in tissue, light is multiply scattered and the spatial resolution is limited by acoustics. As higher acoustic frequencies, which have smaller wavelengths and allow a better resolution, are stronger damped than lower frequencies, the spatial resolution decreases with depth. The spatial resolution is limited at such depths by the acoustic diffraction limit that corresponds to the highest detectable frequency. The ratio of the imaging depth to the best spatial resolution is roughly a constant of 200 \cite{Wang2012}. Only recently published non-linear imaging methods, which use additional information such as sparsity of the imaged structure, can overcome the acoustic diffraction limit and are therefore called "super-resolution" \cite{Murray2017, Hojman2017}. Technical limitations, such as a bandwidth mismatch between the acoustic transducer and the acoustic signal on the sample surface or the noise of an amplifier can reduce the resolution in addition. 

There have been several attempts for mathematically compensating the acoustic attenuation to get images with a higher spatial resolution. Already in 2005, La Rivi\`{e}re et al.\ proposed an integral equation that relates the measured acoustic signal at a given transducer location in the presence of attenuation to the ideal signal in the absence of attenuation \cite{Riviere2005, LaRiviere2006}. Ammari et al.\ later gave a compact derivation of this integral equation directly using the wave equations, which is valid for all dimensions \cite{Ammari2012}. This implies that compensation of acoustic attenuation and dispersion in two or three dimensions can always be reduced to a one-dimensional problem in a two-stage process: first, for each detector location the ideal signal in the absence of attenuation is calculated from the measured signal. This is a one-dimensional reconstruction. In a second step, any reconstruction method for photoacoustic tomography can be used for reconstructions in higher dimensions \cite{Riviere2005, LaRiviere2006}. De\`{a}n-Ben et al.\ described the effects of acoustic attenuation (amplitude reduction and signal broadening), compared the effects of attenuation to the influence of the transducer bandwidth and space-dependent speed of sound and established a correction term similar to La Rivi\`{e}re, but for space-dependent attenuation \cite{Dean-Ben2011}. Kowar and Scherzer used a similar formulation for other lossy wave equations \cite{Kowar2012}. 

Burgholzer et al.\ have compensated directly the attenuation in photoacoustic tomography by using a time reversal finite differences method with a lossy wave equation  \cite{Burgholzer2007, PBurgholzer2007, PBurgholzer2010, PBurgholzer2010a, PBurgholzer2012}. Time reversal of the attenuation term causes the acoustic waves in the finite differences model to grow, as they propagate back in time through the tissue. At each time step the total acoustic energy is controlled by cutting high frequency signals, which would otherwise grow too quickly. This approach was later extended by Treeby et al.\ to account for general power law absorption behavior \cite{Treeby2010, Treeby2010a}. Inspired by attenuation compensation in seismology Treeby proposed a new method for attenuation compensation in photoacoustic tomography using time-variant filtering \cite{Treeby2013}.  

All these attempts have in common that the compensation of the frequency-dependent attenuation is an ill-posed problem, which needs regularization. The physical reason for this ill-posedness is thermodynamics: acoustic attenuation is an irreversible process and the entropy production, which is the dissipated energy of the attenuated acoustic wave divided by the temperature, is equal to the information loss for the reconstructed image \cite{PBurgholzer2012}. This limits also spatial resolution, which correlates with the information content of the reconstructed image. To reach this thermodynamic resolution limit for compensation of acoustic attenuation it is necessary to measure the broadband ultrasonic attenuation parameters of tissues or liquids very accurately \cite{Bauer-Marschallinger2012} and to evaluate the existing mathematical models to get an accurate description of attenuation \cite{Roitner2012}.

The first and most basic description of an attenuated acoustic wave has been given already in 1845 by Stokes \cite{stokes_2009} for fluids. It is based on the assumption that in the presence of attenuation, density changes in the fluid do not react immediately to pressure changes, but only with some relaxation time $\tau$. If $\tau$ is further expressed in terms of viscosity and specific heat, this equation is also known as the thermoviscous equation, which describes approximately a quadratic increase of attenuation with frequency and describes attenuation in liquids very well. Stokes' equation is not only causal in the sense, that it satisfies Kramers-Kronig equation, but also satisfies, as shown by Buckingham, a stronger causality condition: everywhere the predicted pressure pulse is maximally flat at the instant the source is activated: the pressure and all its time derivatives are identical to zero at the origin of time \cite{Buckingham2005}. 
%To adequately model other attenuation laws different from a quadratic increase of attenuation with frequency, Szabo \cite{Szabo1994} and Sushilov and Cobbold \cite{Sushilov2004} have proposed wave equations with a complex wavenumber $K(\omega) = \frac{\omega}{c(\omega)} + i \alpha(\omega)$ providing a general frequency dependent attenuation $\alpha(\omega)$, such as a power law with an exponent different from two. E.g., La Rivi\`{e}re used an exponent of one \cite{Riviere2005}. Other wave equations, which are able to describe a power-law frequency dependence of attenuation, were proposed by Treeby and Cox \cite{Treeby2010} using a fractional Laplacian, and by Nachman et al. assuming $N$ different relaxation processes \cite{Nachman1990}. Roitner and Burgholzer have used that for every wave equation describing attenuation instead of the complex wavenumber $K(\omega)$ and a real frequency $\omega$ it is also possible to take a real wavenumber and a complex frequency \cite{Roitner2011}. The pressure wave is decomposed into plane waves, which show an exponentially decreasing amplitude in time (viewpoint of a standing wave in a resonator exhibiting energy loss). This different viewpoint gives mathematically the same solutions for the wave equation, but the regularization appears already in a diagonal form and is therefore computationally simpler \cite{Roitner2011}.

In the proposed approach, the linear partial differential equation (PDE), which describes the propagation and attenuation of the acoustic wave, is discretized to bring it into the form of a linear state space model (SSM). This can be done for any of the wave equations describing acoustic attenuation, as long as the source term containing the heating function, which describes the absorbed optical energy density per unit time deposited at a certain depth, can be written as a product of a space and a time dependent function. In this work Stokes' equation is used as being the first one and just one attenuation term has to be added compared to the wave equation without attenuation \cite{Buckingham2005}:
\begin{equation}
\frac{\partial^2 p(z,t)}{\partial z^2} - \frac{1}{c_0^2} \frac{\partial^2 p(z,t)}{\partial t^2} + \tau \frac{\partial^3 p(z,t)}{\partial t \partial z^2} = - \frac{\beta}{C_p} \frac{\partial H(z,t)}{\partial t}, \label{equ:RECENDT_001}
\end{equation}
where $p(z,t)$ is the local pressure at depth $z$ and at time $t$, $c_0$ is the ultrasound wave velocity, $\tau$ is the relaxation time, $\beta$ denotes the coefficient of thermal expansion, and $C_p$ the specific heat. The heating function $H(z,t)=R(z)i(t)$ is the product of the fractional energy absorption $R(z)$ at depth $z$ and the temporal profile of the illumination $i(t)$ \cite{Riviere2005}. 

Eq.~\eqref{equ:RECENDT_001} is a one-dimensional equation, which covers all the irreversibility of acoustic attenuation. The integral equation \cite{Riviere2005, LaRiviere2006, Ammari2012} that relates the measured acoustic signal at a given transducer location in the presence of attenuation to the ideal signal in the absence of attenuation is the same for all dimensions and therefore it is sufficient to take the one-dimensional equation \eqref{equ:RECENDT_001}. Two- or three-dimensional reconstructions can be performed in a two-stage process as mentioned above.

The main idea of the proposed approach is the following:
\begin{enumerate}
\item Discretize the PDE in \eqref{equ:RECENDT_001} and bring it into the form of a linear SSM. We do this in a form such that the measurements are linearly connected with a vector that is related to the absorption profile. %Consequently, the next step is:
\item Apply linear estimators on the measurements to estimate this vector. During the discretization, we will show that this vector origins from a non-linear transformation of the absorption profile. %However, this non-linear transformation fulfills some properties allowing for:
\item Estimate the absorption profile based on the estimated vector. 
\end{enumerate} 

We will show that our method observes the following features:
\begin{enumerate}[a)]
\item it allows for probes with arbitrary absorption profiles;
\item it accounts for frequency dependent attenuation of the ultrasound waves;
\item the decrease in laser intensity because of absorption is incorporated;
\item the laser modulation signal can be completely arbitrary and it is not constrained to signals with a well-behaving autocorrelation function;
\item structural properties of the model such as stability, observability and controllability can be easily verified. 
\end{enumerate} 
Besides the fact that the proposed estimation method allows to estimate the absorption profile considering all mentioned effects, it furthermore allows to optimize the laser modulation signal such that the accuracy of the estimated absorption profile is maximized. We will show how this optimization is performed and we will demonstrate the improvement in accuracy by utilizing the optimized laser modulation signals. It will turn out that the reconstruction accuracy can significantly be increased by utilizing the optimized laser modulation signal.

%\color{red} Thomas: You've sent us some papers with other reconstruction approaches that also use some matrix formulation (Mail from 20.1.2017; 10:44). Maybe you could include a phrase concerning the state of the art w.r.t. these papers? Maybe some arguments on how they differ from our approach... as you wish. \color{black}

%Notation:
%\\ 
In what follows the lower-case bold face variables ($\ve{a}$, $\ve{b}$,...) indicate vectors, and the upper-case bold face variables ($\m{A}$, $\m{B}$,...) indicate matrices. 
%We further use $\mathbb{R}$ and $\mathbb{C}$ to denote the set of real and complex numbers, respectively, 
We further use $(\cdot)^T$ to denote transposition, $\m{I}^{n\times n}$ to denote the identity matrix of size $n\times n$, and $\m{0}^{m\times n}$ to denote the zero matrix of size $m\times n$. If the dimensions are clear from context we simply write $\m{I}$ and $\m{0}$, respectively. 
%The real and imaginary part of a variable are indicated by $\re{\cdot}$ and $\im{\cdot}$, respectively. 
%$E[\cdot]$ denotes the expectation operator. In most of the cases we use an index to denote the averaging PDF, however, if the averaging PDF is clear from context, the index is sometimes omitted.

% -------------------------------------------------------------------
% Discretization
% -------------------------------------------------------------------
\section{Discretization of the PDE}
\label{sec:discretization}

In this section, the workflow of deriving a discrete SSM that approximates the physical processes according to Stokes' equation in \eqref{equ:RECENDT_001} is described. This is done by utilizing finite differences in a 1D space. Furthermore, the discretization is performed in a way such that the unknown absorption profile is isolated in a vector, which can be estimated in a follow-up step. We begin with some notational definitions. 

We assume the 1D probe begins at $z=0$. The z-axis is divided into $N_\ve{z}$ equally spaced elements. The width of each element is denoted by $\Delta_z$ and the left border of each element is located at $z_n = n \Delta_z$ with $n = 0, \hdots, N_\ve{z}-1$. The locations $z_n$ are referred to as grid points. The vector $\ve{z} \in \mathbb{R}^{N_\ve{z} \times 1}$ is defined as the vector containing all grid points $z_n$ for $n = 0, \hdots, N_\ve{z}-1$. The function $p(z,t)$ in \eqref{equ:RECENDT_001} describes the local pressure at location $z$ and at time $t$. Based on that, we define the vector $\ve{p}_k \in \mathbb{R}^{N_\ve{z} \times 1}$ as the local pressure at all grid points in $\ve{z}$ at the time $t = k \Delta_t$, where  $\Delta_t$ is the step width of the time discretization.

% Note that some of the following definitions may change at the end of Sec.~\ref{sec:Linear_Model_Formulation} when reflections of the ultrasound waves are considered.  

We now turn to the first term in \eqref{equ:RECENDT_001}. The second derivative of $p(z,t)$ w.r.t. $z$ can be approximated using the central finite difference of second order given by
\begin{align}
\frac{\partial^2 p(z,t)}{\partial z^2} \approx \frac{1}{\Delta_z^2} \left(  p(z_{n-1},t) - 2 p(z_{n},t) + p(z_{n+1},t) \right). \label{equ:RECENDT_004}
\end{align}
For the time $t=k \Delta_t$, the right hand side of \eqref{equ:RECENDT_004} can be written as a vector matrix product according to
\begin{align}
\frac{\partial^2 p(z,t)}{\partial z^2} \rightarrow \m{D} \ve{p}_k, \label{equ:RECENDT_005}
\end{align}
where the matrix $\m{D}\in \mathbb{R}^{N_\ve{z} \times N_\ve{z}}$ is given by
\begin{align}
 \m{D} = \frac{1}{\Delta_z^2} \begin{bmatrix}
 -2      & 1     & 0      & 0      & 0      & \hdots & 0 \\
 1      & -2     & 1      & 0      & 0      & \hdots & 0 \\
 0      &  1     & -2     & 1      & 0      & \hdots & 0 \\ 
 \vdots &        & \ddots & \ddots & \ddots &        & \vdots \\
 0      & \hdots & 0      & 1      & -2     & 1     & 0 \\
 0      & \hdots & 0      & 0      & 1      & -2     & 1 \\
 0      & \hdots & 0      & 0      & 0      & 1     & -2
\end{bmatrix}.  \label{equ:RECENDT_006}
\end{align}

The second term in \eqref{equ:RECENDT_001} can be discretized and written in terms of $\ve{p}_k$ as
\begin{align}
-\frac{1}{c_0^2} \frac{\partial^2 p(z,t)}{\partial t^2} \rightarrow -\frac{1}{c_0^2 \Delta_t^2} \left( \ve{p}_{k+1} - 2 \ve{p}_k + \ve{p}_{k-1} \right), \label{equ:RECENDT_007}
\end{align}
which corresponds to the second order central finite difference. Similarly, the third term in \eqref{equ:RECENDT_001} can be discretized using \eqref{equ:RECENDT_005} as follows:
\begin{align}
\tau \frac{\partial^3 p(z,t)}{\partial t \partial z^2} \rightarrow \frac{\tau}{2 \Delta_t}\m{D} \left( \ve{p}_{k+1} - \ve{p}_{k-1} \right) \label{equ:RECENDT_008}\,.
\end{align}

For the right hand side of \eqref{equ:RECENDT_001}, we utilize $H(z,t) = R(z) i(t)$ \cite{Petschke_2010}, where $R(z)$ accounts for the absorption and where $i(t)$ is the laser intensity or laser modulation function. For a homogeneous medium, it holds that $R(z) = \chi \mu \mathrm{e}^{-\mu z} $, where $\mu$ is the absorption coefficient of the laser light and $\chi $ is the fluence of the laser light at the sample surface. Decrease in irradiance of the laser intensity is considered via the term $\mathrm{e}^{-\mu z}$ within $R(z)$. We consider an inhomogeneous probe $\mu(z)$, which can be discretized as $\mu_n = \mu(z_n)$ for $n = 0,\hdots, N_\ve{z}-1$. The vector $\bm{\mu} \in \mathbb{R}^{N_\ve{z} \times 1}$ with the elements $\mu_n$ describes the discretized absorption profile of the probe. The term $a_n = \mathrm{e}^{-\mu_n \Delta_z}$ approximately describes the attenuation of the laser intensity between the grid points $z_n$ and $z_{n+1}$. Let $i_k$ denote the continuous laser intensity $i(t)$ at time $t = k \Delta_t$ at the surface of the probe ($z=0$) and let $\tilde{i}_k(z_n)$ denote the laser intensity at time $t = k \Delta_t$ and at the $n^{\text{th}}$ grid point $z_n$, then we have
\begin{align}
\tilde{i}_k(z_n) \approx a_{n-1} a_{n-2} \hdots a_1 a_0 i_k.  \label{equ:RECENDT_011}
\end{align}
$\tilde{i}_k(z_n)$ evaluated for every grid point $z_n$ for $n = 0,\hdots, N_\ve{z}-1$ can be written in vector form as $\tilde{\ve{i}}_k \in \mathbb{R}^{N_\ve{z} \times 1}$ by
\begin{align}
\tilde{\ve{i}}_k \approx \begin{bmatrix}
1 \\ a_0 \\ a_1 a_0 \\ \vdots \\ a_{N_\ve{z}-2} a_{N_\ve{z}-3} \hdots a_1 a_0
\end{bmatrix} i_k. \label{equ:RECENDT_012}
\end{align}
We are now able to evaluate $H(z,t) = R(z) i(t)$ for every grid point at time $t = k \Delta_t$, which is denoted by $\ve{h}_k \in \mathbb{R}^{N_\ve{z} \times 1}$. Then, $\ve{h}_k$ follows from \eqref{equ:RECENDT_012} by multiplying every element of $\tilde{\ve{i}}_k$ with $\chi$ and the corresponding $\mu_n$, yielding
\begin{align}
\ve{h}_k \approx & \, \chi \underbrace{\begin{bmatrix}
\mu_0 \\ \mu_1 a_0 \\ \mu_2 a_1 a_0 \\ \vdots \\  \mu_{N_\ve{z}-1} a_{N_\ve{z}-2} a_{N_\ve{z}-3} \hdots a_1 a_0
\end{bmatrix}}_{\ve{d}} i_k \label{equ:RECENDT_012a} \\
=& \, \chi \ve{d} i_k. \label{equ:RECENDT_013}
\end{align}
With this result, the right hand side of \eqref{equ:RECENDT_001} follows to
\begin{align}
-\frac{\beta}{C_p} \frac{\partial H(z,t)}{\partial t} \rightarrow &  -\frac{\beta}{C_p} \frac{\partial \ve{h}_k}{\partial t} \label{equ:RECENDT_009} \\
\approx& \underbrace{-\frac{\beta \chi}{C_p \Delta_t} \ve{d}}_{\ve{b}} \underbrace{\left( i_k - i_{k-1} \right)}_{u_k} \label{equ:RECENDT_014} \\
=& \ve{b} u_k, \label{equ:RECENDT_015} 
\end{align}
where the backward difference was used to approximate $\partial \ve{h}_k / \partial t$ in \eqref{equ:RECENDT_009}.

Now, the discretized form of the PDE in \eqref{equ:RECENDT_001} follows by combining \eqref{equ:RECENDT_005}, \eqref{equ:RECENDT_007}, \eqref{equ:RECENDT_008} and \eqref{equ:RECENDT_015} as
\begin{align}
& \m{D} \ve{p}_k \hspace{-1pt}  - \hspace{-1pt} \frac{1}{c_0^2 \Delta_t^2} \left( \ve{p}_{k+1} - 2 \ve{p}_k + \ve{p}_{k-1} \right) +  \frac{\tau}{2 \Delta_t}\m{D} \left( \ve{p}_{k+1} \hspace{-1pt} -  \ve{p}_{k-1} \right)  \nonumber \\
&= \underbrace{\left(-\frac{1}{c_0^2 \Delta_t^2} \m{I} + \frac{\tau}{2 \Delta_t}\m{D} \right)}_{\m{M}_1} \ve{p}_{k+1} + \underbrace{ \left( \m{D} + \frac{2}{c_0^2 \Delta_t^2}  \m{I} \right)}_{\m{M}_2} \ve{p}_k   + \underbrace{\left( -\frac{1}{c_0^2 \Delta_t^2}  \m{I}  - \frac{\tau}{2 \Delta_t}\m{D} \right)}_{\m{M}_3} \ve{p}_{k-1}  \label{equ:RECENDT_016}  \\
& = \m{M}_1 \ve{p}_{k+1} + \m{M}_2 \ve{p}_k + \m{M}_3 \ve{p}_{k-1} =  \ve{b} u_k, \label{equ:RECENDT_017} 
\end{align}
and further
\begin{align}
  \ve{p}_{k+1} =& \underbrace{-\m{M}_1^{-1} \m{M}_2}_{\m{M}_4} \ve{p}_k + \underbrace{\left(-\m{M}_1^{-1} \m{M}_3\right)}_{\m{M}_5} \ve{p}_{k-1} + \underbrace{\m{M}_1^{-1}\ve{b}}_{\ve{f}} u_k \label{equ:RECENDT_017b} \\
  =& \m{M}_4 \ve{p}_k + \m{M}_5 \ve{p}_{k-1} + \ve{f} u_k. \label{equ:RECENDT_017a} 
\end{align}
This result shows that for approximating the pressure profile $\ve{p}_{k+1}$, the current pressure profile $\ve{p}_k$ as well as the previous one $\ve{p}_{k-1}$ are required. In order to bring this equation into the form of an SSM, we define the state vector
\begin{align}
\ve{x}_k = \begin{bmatrix}
\ve{p}_k \\ \ve{p}_{k-1}
\end{bmatrix}, \label{equ:RECENDT_018} 
\end{align}
which allows to bring \eqref{equ:RECENDT_017a} into the form of
\begin{align}
\ve{x}_{k+1} = \begin{bmatrix}
\ve{p}_{k+1} \\ \ve{p}_{k}
\end{bmatrix} =& \underbrace{\begin{bmatrix}
\m{M}_4 & \m{M}_5 \\
\m{I} & \m{0} 
\end{bmatrix}}_{\m{A}} \ve{x}_{k} + \underbrace{\begin{bmatrix} \ve{f} \\ \ve{0} \end{bmatrix}}_{\ve{g}}  u_k \label{equ:RECENDT_019} \\
=& \m{A} \ve{x}_{k} + \ve{g} u_k, \label{equ:RECENDT_020}
\end{align}
which represents the final form of the discretized Stokes' equation in \eqref{equ:RECENDT_001}.

The next step for deriving an SSM representation of \eqref{equ:RECENDT_001} is to develop the measurement equation. The measurement at time instance $t = k \Delta_t$ is denoted by $y_k$ and it is given by the pressure at the surface of the probe plus some additive measurement noise $w_k$, according to
\begin{align}
y_k =& \, p(z=0, t = k \Delta_t) + w_k  \label{equ:RECENDT_021} \\
=&\, \ve{c}^T \ve{x}_{k}  + w_k,  \label{equ:RECENDT_022}
\end{align}
where $\ve{c}^T \in \mathbb{R}^{1 \times 2N_\ve{z}}$ is a row vector with a $1$ at its first entry and all zeros elsewhere. The measurement noise $w_k$ in \eqref{equ:RECENDT_022} is assumed to be zero mean white Gaussian noise with variance $\sigma_w^2$. Combining \eqref{equ:RECENDT_020} and \eqref{equ:RECENDT_022} forms the final expression for the SSM
\begin{align}
\ve{x}_{k+1} &= \m{A} \ve{x}_{k} + \ve{g} u_k + \ve{q}_k \label{equ:RECENDT_002} \\
y_k &= \ve{c}^T \ve{x}_{k} + w_k. \label{equ:RECENDT_003}
\end{align}
Note that in \eqref{equ:RECENDT_002} we added $\ve{q}_k$ as an additional noise term which shall account for model errors due to the discretization process. We will neglect this noise term in the following section for simplicity and clarity. However, in Section~\ref{sec:experiments} we will assume that the upper half of $\ve{q}_k$ is a zero mean Gaussian random vector with covariance matrix $\sigma_q^2 \m{I}^{N_\ve{z} \times N_\ve{z}}$ and the lower half of $\ve{q}_k$ is zero. 

We draw the attention to the fact that in \eqref{equ:RECENDT_003} the variables $y_k$'s are assumed to be pressure measurements. In cases where measurements are given in another form, such as the deflection of the surface of the probe, \eqref{equ:RECENDT_003} needs to be adapted accordingly.
		
Also note that in our model the only unknown material parameter is the absorption profile $\bm{\mu}$, which is related to the vector $\ve{d}$ via \eqref{equ:RECENDT_012a}. Besides $\bm{\mu}$, also the pressure profile $\ve{p}_k$ inside the state vector $\ve{x}_{k}$ is unknown for all values of $k$. In Section~\ref{sec:Linear_Model_Formulation}, we will rewrite the SSM in \eqref{equ:RECENDT_002} and \eqref{equ:RECENDT_003} such that the vector $\ve{d}$ is linearly connected with the measurements. We will show that the vector $\ve{d}$ can be linearly estimated without the need of estimating the pressure profile $\ve{p}_k$ or the state vector $\ve{x}_{k}$ at any time. For the pressure waves modelled with \eqref{equ:RECENDT_002}, reflections at the boundaries of the simulation area are observed. If these reflections do not meet with the properties of underlying physical processes of the probe, then there are two common options to prevent them:
\begin{itemize}
\item introduce a perfectly matched layer at the boundaries of the simulation area;
\item make the simulation area larger such that the reflected waves do not disturb the measurements. 
\end{itemize}
In this work, we implement the latter solution and thus we refer to the unaltered simulation area as \emph{area of interest}.

% -------------------------------------------------------------------
% Stability of the SSM
% -------------------------------------------------------------------
\section{Properties of the state space model}
\label{sec:Stability}
In control system design, stability is a fundamental requirement, which describes the properties of the equilibrium points in the state space. There are various interpretations of the stability 
such as convergence to an equilibrium, the speed of the convergence, boundedness of the input, the output and the state. In this section, we give a necessary and sufficient condition for the asymptotic stability of the noise free SSM, i.e.\ $\ve{q}_k=\ve{0},\;w_k=0$ in \eqref{equ:RECENDT_002}-\eqref{equ:RECENDT_003}, and show that the construction of this model is well-defined. Additionally, we prove that the state space representation of \eqref{equ:RECENDT_001} is observable, which means that it is possible to determine any arbitrary initial state $\m{x}_0$ from observing a finite sequence of output variables $y_k\,(0\leq k \leq m)$. Finally, we show that any desired final state can be reached from any initial state by using the proper input signal $u_k$, i.e. the system is controllable.

The proposed discrete SSM is also a linear and time-invariant (LTI) system, which is asymptotically stable if and only if all the eigenvalues of the state matrix $\m{A}$ lie inside the unit circle. In order to examine this property, we first prove some identities about the eigenvalues of $\m{A}$.

\begin{lemma}
\label{lemma:eig_Mi}
The eigenvalues of the matrices $\m{M}_i\in\IR^{N_{\m{z}} \times N_{\m{z}}}\; (i=1,2,3)$ in \eqref{equ:RECENDT_016} can be given in explicit forms:
\begin{equation}
		\lambda_k(\m{M}_{1,3})=\pm \gamma \lambda_k(\m{D}) - \alpha, \quad \lambda_k(\m{M}_2)=\lambda_k(\m{D}) + 2 \alpha, \qquad (1\leq k \leq N_{\m{z}})\,, 
\label{eq:eig_Mi}
\end{equation}
where $\lambda_k(\m{D})<0$ denotes the eigenvalues of $\m{D}$, $\gamma=\dfrac{\tau}{2\Delta_t}>0$ and $\alpha=\dfrac{1}{c_0^2\Delta_t^2}>0$.
\end{lemma}

\begin{proof} 
Note that the second order finite difference matrix $\m{D}$ is also a symmetric tridiagonal Toeplitz matrix. According to Section~3 in \cite{diff_eig_siam}, the eigenvalues of $\m{D}$ and the corresponding eigenvectors are %According to Proposition 2.1 in \cite{toeplitz_eig}, the eigenvalues of $\m{D}$ and the corresponding eigenvectors are
\begin{align}
	\label{eq:eig_D}
	\lambda_k(\m{D})&=-\frac{1}{\Delta^2_z}\cdot\left(2-2\cos\left(\frac{k\pi}{N_{\m{z}}+1}\right)\right) & (1 \leq k \leq N_{\m{z}})\,,\\
	\m{v}_k\left[j\right]&=\sqrt{\frac{2}{N_{\m{z}}+1}} \cdot \sin \left(\frac{kj\pi}{N_{\m{z}}+1}\right) & (1 \leq k \leq N_{\m{z}})\,,
	\label{eq:eigv_D}
\end{align}
where $\lambda_k(\m{D})$ are always strictly negative. Now let us consider the eigenvalues $\lambda_k(\m{D})$ and the corresponding eigenvectors $\m{v}_k$. Then we have
\begin{equation}
	\m{M}_1 \m{v}_k=(\gamma \m{D} - \alpha \m{I})\m{v}_k=\gamma \m{D} \m{v}_k - \alpha \m{v}_k=\gamma \lambda_k(\m{D}) \m{v}_k  -\alpha \m{v}_k=(\gamma \lambda_k(\m{D})  -\alpha) \m{v}_k\,,
\end{equation}  
hence $(\gamma \lambda_k(\m{D}) - \alpha,\;\m{v}_k)$ are eigenpairs of $\m{M}_1$ for $1\leq k \leq N_{\m{z}}$. The proof is analogous for $\m{M}_2$ and $\m{M}_3$.
\end{proof} 
\begin{corollary}
\label{cor:negdef}
If $c_0,\tau,\Delta_t,\Delta_z\in\IR_+$, $\m{M}_1$ is negative definite, and thus it is invertible. Therefore, the state matrix $\m{A}$ in \eqref{equ:RECENDT_019} is well-defined.
%According to \eqref{eq:eig_D}, all the eigenvalues of $\m{M_1}$ are strictly negative, thus it is negative definite.  
\end{corollary}

\begin{lemma}
\label{lemma:eig_A}
If the state matrix $\m{A}$ in \eqref{equ:RECENDT_017b} is invertible, then its eigenvalues can be written as
\begin{equation}
		\lambda_{2k-1,2k}(\m{A})=\frac{-\lambda_k(\m{M}_2) \pm \sqrt{\lambda^2_k(\m{M}_2)-4\cdot\lambda_k(\m{M}_1)\lambda_k(\m{M}_3)}}{2\cdot\lambda_k(\m{M}_1)} \qquad (1\leq k \leq N_{\m{z}})\,. 
\label{eq:eig_A}
\end{equation}
\end{lemma}
\begin{proof} 
In Lemma~\ref{lemma:eig_Mi} we showed that the matrices $\m{M}_i\;(i=1,2,3)$ have the same eigenvectors $\m{v}_k\;(1\leq k\leq N_{\m{z}})$. In addition, $\m{M}_i$'s are symmetric and so the matrix $\m{V}=(\m{v}_1,\ldots,\m{v}_{N_{\m{z}}})$ is orthogonal. Therefore, they can be diagonalized as follows 
\begin{equation*}
	\m{M}_i=\m{V} \m{\Lambda}_i \m{V}^T\qquad (i=1,2,3)\,,
\end{equation*}  
where $\m{\Lambda}_i=\diag(\lambda_1(\m{M}_i),\ldots,\lambda_{N_{\m{z}}}(\m{M}_i))$. Furthermore, by \eqref{equ:RECENDT_017b} we have that 
\begin{equation*}
	\m{M}_4=-\m{M}_1^{-1} \m{M}_2=-\m{V} \m{\Lambda}_1^{-1} \m{ \Lambda}_2 \m{V}^T\,,\qquad \m{M}_5=-\m{M}_1^{-1}\m{M}_3=-\m{V} \m{\Lambda}_1^{-1} \m{\Lambda}_3 \m{V}^T\,. 
\end{equation*}  
Using these identities, the state space model can be transformed into an equivalent form
\begin{align}
	\label{eq:mod_ssm1}
	\m{\widetilde{x}}_{k+1}&=\m{\widetilde{A}}\m{\widetilde{x}}_k+\m{\widetilde{g}}u_k\,, \\
	\label{eq:mod_ssm2}
	\m{y}_k&=\m{\widetilde{c}}^T\m{\widetilde{x}}_k\,,	
\end{align}  
where $\ve{\widetilde{x}}_k=\m{T}\m{x}_k$, $\m{\widetilde{g}}=\m{Tg}$, $\m{\widetilde{c}}=\m{Tc}$, and $\m{T}$ is an orthogonal matrix: 
\begin{equation}
\renewcommand{\arraystretch}{1.5}
	\m{T}=
	\begin{bmatrix}
    \m{V}^T & \m{0}\\
    \m{0} & \m{V}^T
\end{bmatrix}\,,
	\qquad
	\m{\widetilde{A}}=\m{TAT}^T=
	\begin{bmatrix}
    -\m{\Lambda}_1^{-1} \m{\Lambda}_2 & -\m{\Lambda}_1^{-1} \m{\Lambda}_3\\
    \m{I} & \m{0}
	\end{bmatrix}\,.
\label{eq:trans_ssm}
\end{equation}
The matrices $\m{A}$ and $\m{\widetilde{A}}$ share the same eigenvalues, which satisfy the following equation
\begin{equation}
\renewcommand{\arraystretch}{1.5}
	\m{\widetilde{A}}\m{\widetilde{w}}_k=
	\begin{bmatrix}
    -\m{\Lambda}_1^{-1} \m{\Lambda}_2 & -\m{\Lambda}_1^{-1} \m{\Lambda}_3\\
    \m{I} & \m{0}
	\end{bmatrix}
	\begin{bmatrix}
    \m{\widetilde{w}}_k^{(1)}\\ 
    \m{\widetilde{w}}_k^{(2)}
	\end{bmatrix}
=\lambda_k(\m{\widetilde{A}})\cdot
	\begin{bmatrix}
    \m{\widetilde{w}}_k^{(1)}\\ 
    \m{\widetilde{w}}_k^{(2)}
	\end{bmatrix}
=\lambda_k(\m{\widetilde{A}})\cdot \m{\widetilde{w}}_k\,,
\label{eq:eig_Atilde}
\end{equation}
where $\m{\widetilde{w}}_k\;(1\leq k \leq 2N_{\m{z}})$ denotes the eigenvectors of $\m{\widetilde{A}}$. Provided that $\m{A}$ is invertible, $\lambda_k(\m{\widetilde{A}})\neq 0\,(1\leq k \leq 2N_{\m{z}})$, thus $\m{\widetilde{w}}_k^{(2)}=\dfrac{1}{\lambda_k(\m{\widetilde{A}})}\m{\widetilde{w}}_k^{(1)}$ and by substitution we have 
\begin{equation*}
	-\m{\Lambda}_1^{-1} \m{\Lambda}_2\m{\widetilde{w}}_k^{(1)} - \frac{1}{\lambda_k(\m{\widetilde{A}})} \m{\Lambda}_1^{-1} \m{\Lambda}_3 \m{\widetilde{w}}_k^{(1)}=\lambda_k(\m{\widetilde{A}}) \cdot \m{\widetilde{w}}_k^{(1)}\,.
\end{equation*}
After rearranging and multiplying both sides by $\lambda_k(\m{\widetilde{A}})\cdot \m{\Lambda}_1$ we get the following matrix equation
\begin{equation}
	\m{\widetilde{\Lambda}}\m{\widetilde{w}}_k^{(1)}=\left(\lambda_k^2(\m{\widetilde{A}})\m{\Lambda}_1 + \lambda_k(\m{\widetilde{A}}) \m{\Lambda}_2 + \m{\Lambda}_3\right) \m{\widetilde{w}}_k^{(1)}=\m{0}\qquad (1\leq k \leq 2N_{\m{z}})\,.
	\label{eq:matrixeq}
\end{equation}
Note that $\m{\widetilde{\Lambda}}$ is a diagonal matrix, for which every diagonal element is a quadratic polynomial in $\lambda_k(\m{\widetilde{A}})$. Therefore, the eigenvalues are equal to the roots of these polynomials:
\begin{equation}
\lambda_{2k-1,2k}(\m{\widetilde{A}})=\frac{-\lambda_k(\m{\Lambda}_2) \pm \sqrt{\lambda_k^2(\m{\Lambda}_2) - 4\cdot \lambda_k(\m{\Lambda}_1) \lambda_k(\m{\Lambda}_3})}{2\lambda_k(\m{\Lambda}_1)} \qquad	(1\leq k \leq N_{\m{z}})\,.
\label{eq:quad_pol}
\end{equation} 
The statement of the Lemma follows from the similarity of the matrices $\m{A}$, $\m{\widetilde{A}}$ and $\m{M}_i$, $\m{\Lambda}_i\;(i=1,2,3)$.
\end{proof} 

\begin{lemma}
\label{lemma:eigvecs_A}
If the eigenvalues $\lambda_k(\m{\widetilde{A}})\; (k=1,\ldots,2N_{\m{z}})$ are pairwise distinct and non-zero then $\m{\widetilde{A}}$ in the transformed SSM \eqref{eq:mod_ssm1} can be diagonalized as follows
\begin{equation*}
	\m{\widetilde{A}}=\m{\widetilde{W}}\m{\Lambda}\m{\widetilde{W}}^{-1}=
	\begin{bmatrix*}[l]
    \phantom{.}\m{I} & \phantom{.}\m{I}\\ 
    \m{\Lambda}_{+}^{-1} &\m{\Lambda}_{-}^{-1}
	\end{bmatrix*}
	\begin{bmatrix*}[l]
    \m{\Lambda_{+}} & \m{0}\\ 
    \m{0} &\m{\Lambda_{-}}
	\end{bmatrix*}
		\begin{bmatrix}
       \m{\Lambda}_{+}\left(\m{\Lambda}_{+}-\m{\Lambda}_{-}\right)^{-1} & \m{\Lambda}_{+} \m{\Lambda}_{-} \left(\m{\Lambda}_{-}-\m{\Lambda}_{+}\right)^{-1}\\
       \m{\Lambda}_{-}\left(\m{\Lambda}_{-}-\m{\Lambda}_{+}\right)^{-1} & \m{\Lambda}_{+} \m{\Lambda}_{-} \left(\m{\Lambda}_{+}-\m{\Lambda}_{-}\right)^{-1}
	  \end{bmatrix}\,,	
		%\label{eq:eigvec_W}
		\end{equation*}
where the matrices $\m{\Lambda_{-}},\,\m{\Lambda_{+}}\in\IR^{N_{\m{z}} \times N_{\m{z}}}$ contain the first and the second roots of the quadratic polynomials in \eqref{eq:quad_pol}, namely $\m{\Lambda_{+}}=\diag\left(\lambda_1(\m{\widetilde{A}}),\lambda_3(\m{\widetilde{A}}),\ldots,\lambda_{2 N_{\m{z}} - 1}(\m{\widetilde{A}})\right)$ and $\m{\Lambda_{-}}=\diag\left(\lambda_2(\m{\widetilde{A}}),\lambda_4(\m{\widetilde{A}}),\ldots,\lambda_{2 N_{\m{z}}}(\m{\widetilde{A}})\right)$.
%\begin{equation}
	%\m{\widetilde{W}}=	
		%\begin{bmatrix*}[l]
    %\phantom{.}\m{I} & \phantom{.}\m{I}\\ 
    %\m{\Lambda}_{+}^{-1} &\m{\Lambda}_{-}^{-1}
	%\end{bmatrix*}
	%\qquad
	%\m{\widetilde{Q}}=\m{\widetilde{W}}^{-1}=	
	%%\begin{bmatrix}
    %%\left(\m{\Lambda_{+}}-\m{\Lambda_{-}}\right)^{-1}\m{\Lambda_{+}}& \m{0}\\
    %%\m{0} & \left(\m{\Lambda_{+}}-\m{\Lambda_{-}}\right)^{-1}\m{\Lambda_{+}}
	%%\end{bmatrix}
		%\begin{bmatrix}
       %\phantom{-}\m{\Lambda}_{+}\left(\m{\Lambda}_{+}-\m{\Lambda}_{-}\right)^{-1} & \m{\Lambda}_{+}^{-1}-\m{\Lambda}_{-}^{-1}\\
       %-\m{\Lambda}_{-}\left(\m{\Lambda}_{+}-\m{\Lambda}_{-}\right)^{-1} & \m{\Lambda}_{-}^{-1}-\m{\Lambda}_{+}^{-1}
	  %\end{bmatrix}\,,
		%\label{eq:eigvec_W}
%\end{equation}
%where the matrices $\m{\Lambda_{-}},\,\m{\Lambda_{+}}\in\IR^{N_{\m{z}} \times N_{\m{z}}}$ contain the first and the second roots of the quadratic polynomials in \eqref{eq:quad_pol}, namely $\m{\Lambda_{+}}=\diag\left(\lambda_1(\m{\widetilde{A}}),\lambda_3(\m{\widetilde{A}}),\ldots,\lambda_{2 N_{\m{z}} - 1}(\m{\widetilde{A}})\right)$ and $\m{\Lambda_{-}}=\diag\left(\lambda_2(\m{\widetilde{A}}),\lambda_4(\m{\widetilde{A}}),\ldots,\lambda_{2 N_{\m{z}}}(\m{\widetilde{A}})\right)$.
\end{lemma}

\begin{proof} First, we prove that $\m{\widetilde{A}}\m{\widetilde{W}}=\m{\widetilde{W}} \m{\Lambda}$. By applying Viet\`a's formulas to the quadratic polynomials in \eqref{eq:quad_pol} we have that $-\m{\Lambda}_1^{-1} \m{\Lambda}_2=\m{\Lambda_{-}} + \m{\Lambda_{+}}$ and $-\m{\Lambda}_1^{-1} \m{\Lambda}_3=-\m{\Lambda_{-}} \cdot \m{\Lambda_{+}}$, therefore
\begin{equation}
	\m{\widetilde{A}}\m{\widetilde{W}}=
	%\begin{bmatrix}
    %-\m{\Lambda}_1^{-1} \m{\Lambda}_2 & -\m{\Lambda}_1^{-1} \m{\Lambda}_3\\
    %\m{I} & \m{0}
	%\end{bmatrix}
	%\begin{bmatrix*}[l]
    %\phantom{.}\m{I} & \phantom{.}\m{I}\\ 
    %\m{\Lambda}_{+}^{-1} &\m{\Lambda}_{-}^{-1}
	%\end{bmatrix*}
	%= 	
	\begin{bmatrix}
    \m{\Lambda_{-}} + \m{\Lambda_{+}} & -\m{\Lambda_{-}} \cdot \m{\Lambda_{+}}\\
    \m{I} & \m{0}
	\end{bmatrix}
	\begin{bmatrix*}[l]
    \phantom{.}\m{I} & \phantom{.}\m{I}\\ 
    \m{\Lambda}_{+}^{-1} &\m{\Lambda}_{-}^{-1}
	\end{bmatrix*}
	=
	\begin{bmatrix}
    \m{\Lambda_{+}}& \m{\Lambda_{-}}\\
    \m{I} & \m{I}
	\end{bmatrix}
	%=
	%\begin{bmatrix*}[l]
    %\phantom{.}\m{I} & \phantom{.}\m{I}\\ 
    %\m{\Lambda}_{+}^{-1} &\m{\Lambda}_{-}^{-1}
	%\end{bmatrix*}
	%\begin{bmatrix*}[l]
    %\m{\Lambda}_{+}  & \phantom{.}\m{0}\\ 
    %\phantom{.}\m{0} & \m{\Lambda}_{-}
	%\end{bmatrix*}	
	=\m{\widetilde{W}}\m{\Lambda}
	\,.
\label{eq:eig_eq1}
\end{equation}
If $\m{\widetilde{A}}$ has pairwise distinct eigenvalues then the corresponding eigenvectors, i.e. the columns of $\m{\widetilde{W}}$ are linearly independent, thus it is invertible. The validity of the explicit formula for $\m{\widetilde{W}}^{-1}$ can be verified via simple matrix multiplication: $\m{\widetilde{W}} \m{\widetilde{W}}^{-1}=\m{I}$.
   
%The proof is analogous for showing $\m{\widetilde{Q}}\m{\widetilde{A}}=\m{\Lambda} \m{\widetilde{Q}}$.
%Additionally, the right hand side $\m{\widetilde{W}} \m{\Lambda}$ can be written as follows:
%\begin{equation}
	%\m{\widetilde{W}}\m{\Lambda}&=
	%\begin{bmatrix*}[l]
    %\phantom{.}\m{I} & \phantom{.}\m{I}\\ 
    %\m{\Lambda}_{+}^{-1} &\m{\Lambda}_{-}^{-1}
	%\end{bmatrix*}
	%\begin{bmatrix*}[l]
    %\m{\Lambda}_{+}  & \phantom{.}\m{0}\\ 
    %\phantom{.}\m{0} & \m{\Lambda}_{-}
	%\end{bmatrix*}
	%=
	%\begin{bmatrix}
    %\m{\Lambda_{+}}& \m{\Lambda_{-}}\\
    %\m{I} & \m{I}
	%\end{bmatrix}\,.
%\label{eq:eig_eq1}
%\end{equation}
\end{proof}

\begin{theorem}
\label{theorem:stability}
The state space model in \eqref{equ:RECENDT_002} is asymptotically stable if and only if 
\begin{equation}
	2\cos\left(\frac{k\pi}{N_{\m{z}}+1}\right)-2>-\frac{4\Delta_z^2}{c_0^2 \Delta_t^2}\qquad (1\leq k \leq N_{\m{z}})\,.
\label{eq:stability}
\end{equation}
\end{theorem}
\begin{proof}
In order to prove the stability of a discrete time state space model, one should show that the eigenvalues $\lambda_k(\m{\widetilde{A}})$ lie inside the unit disc. To this end, we apply two steps of the well-known Schur-Cohn algorithm (see Section~6.8 in \cite{henrici}), which provides an equivalent criteria for testing asymptotic stability.

%To this end, a straightforward way would be to analyze the explicit formula in \eqref{eq:quad_pol}. However, it is a tedious and mostly technical computation.
\textbf{Step 1.} Let us consider the diagonal elements of the polynomial matrix in \eqref{eq:matrixeq}, which are quadratic polynomials of the form
\begin{align*}
	P_1(\lambda)=\lambda_k(\m{\Lambda}_1)\lambda^2 +\lambda_k(\m{\Lambda}_2)\lambda + \lambda_k(\m{\Lambda}_3)\,, \quad P_1^*(\lambda)=\lambda_k(\m{\Lambda}_3)\lambda^2 +\lambda_k(\m{\Lambda}_2)\lambda + \lambda_k(\m{\Lambda}_1)\,.
\label{eq:pols_step1}
\end{align*}
for $1\leq k \leq N_{\m{z}}$. By Viet\`a's formulas, the condition $\left|\lambda_k(\m{\Lambda}_3)\right|<\left|\lambda_k(\m{\Lambda}_1)\right|$ should be satisfied, otherwise there is at least one eigenvalue that lies outside the open unit disc. In the notations of \eqref{eq:eig_Mi} and under the assumptions of Corollary~\ref{cor:negdef}, the inequality can be written as
\begin{equation*}
	\left|-\gamma \lambda_k(\m{D}) - \alpha\right|=\left|\lambda_k(\m{\Lambda}_3)\right|<\left|\lambda_k(\m{\Lambda}_1)\right|=\left|\gamma \lambda_k(\m{D}) - \alpha\right|= -\gamma \lambda_k(\m{D}) + \alpha\,.
\end{equation*}
According to the sign on the left hand side, we have the following two cases:  	
	\begin{align*}
		-\gamma \lambda_k(\m{D}) - \alpha &< -\gamma \lambda_k(\m{D}) + \alpha  \quad \Leftrightarrow \quad  0<2\alpha\,,\\
  	 \phantom{-}\gamma \lambda_k(\m{D}) + \alpha &< -\gamma \lambda_k(\m{D}) + \alpha  \quad \Leftrightarrow \quad  0<-2\gamma \lambda_k(\m{D})\,.
	\end{align*}
Since $0<\alpha,\gamma$ and $\lambda_k(\m{D})<0\;(1\leq k \leq N_{\m{z}})$, the condition $\left|\lambda_k(\m{\Lambda}_3)\right|<\left|\lambda_k(\m{\Lambda}_1)\right|$ is satisfied. Now we can proceed by applying the Rouch\`e's theorem on 
\begin{equation}
	\lambda_k(\m{\Lambda}_1)P_1(\lambda)-\lambda_k(\m{\Lambda}_3)P_1^{*}(\lambda)=\left(\lambda^2_k(\m{\Lambda}_1)-\lambda^2_k(\m{\Lambda}_3)\right)\lambda^2+\left(\lambda_k(\m{\Lambda}_1)-\lambda_k(\m{\Lambda}_3)\right) \lambda_k(\m{\Lambda}_2)\lambda\,,
	\label{eq:step1:end}
\end{equation}
which has as many zeros inside the unit disc as $P_1$.

\textbf{Step 2.}
We can simplify \eqref{eq:step1:end} by $\lambda_k(\m{\Lambda}_1)-\lambda_k(\m{\Lambda}_3)\neq0$, then we define the linear polynomials
\begin{align*}
	P_2(\lambda)=\left(\lambda_k(\m{\Lambda}_1)+\lambda_k(\m{\Lambda}_3)\right)\lambda + \lambda_k(\m{\Lambda}_2)\,, \quad P_2^*(\lambda)=\lambda_k(\m{\Lambda}_2) \lambda + \left(\lambda_k(\m{\Lambda}_1)+\lambda_k(\m{\Lambda}_3)\right)\,.
\end{align*}
Again by Viet\`a's formulas, the condition $\left|\lambda_k(\m{\Lambda}_2)\right|<\left|\lambda_k(\m{\Lambda}_1)+\lambda_k(\m{\Lambda}_3)\right|=2\alpha$ should be satisfied. Considering the signs of the left hand side we have
	\begin{align}
		\phantom{-}\lambda_k(\m{D}) + 2\alpha &< 2 \alpha  \quad \Leftrightarrow \quad  \lambda_k(\m{D})<0\,,\\
		\label{eq:step2_last}
		-\lambda_k(\m{D}) - 2\alpha &< 2 \alpha  \quad \Leftrightarrow \quad  \lambda_k(\m{D})>-4\alpha\,.
	\end{align}   
The first inequality is true for all $1\leq k \leq N_{\m{z}}$, hence the eigenvalues of the state matrix lie inside the unit disc if and only if the second condition is satisfied. The statement of the theorem follows by substituting back the definitions of $\lambda_k(\m{D})$ and $\alpha$ into \eqref{eq:step2_last}. 
\end{proof}

In order to demonstrate the results we displayed the eigenvalues of the matrices $\m{D}$ and $\m{A}$ in Fig.~\ref{fig:stability}. One can see that the state matrix has five unstable modes due to the five eigenvalues $\lambda_k(\m{D})$ that violate the condition in \eqref{eq:step2_last}. The other eigenvalues of $\m{A}$ are close to the unit circle, but their absolute values are still less than one. Note that for LTI systems, asymptotic stability is the strongest type of stability, which implies others like Lyapunov stability, bounded input bounded output (BIBO), and bounded input bounded state (BIBS) stability (see Chapter 7. in \cite{bay}).

\begin{figure}[!t]
\centering
  \subfigure[Eigenvalues of $\m{D}$.]{
  \includegraphics[scale=0.58, trim=150 230 160 230, clip]{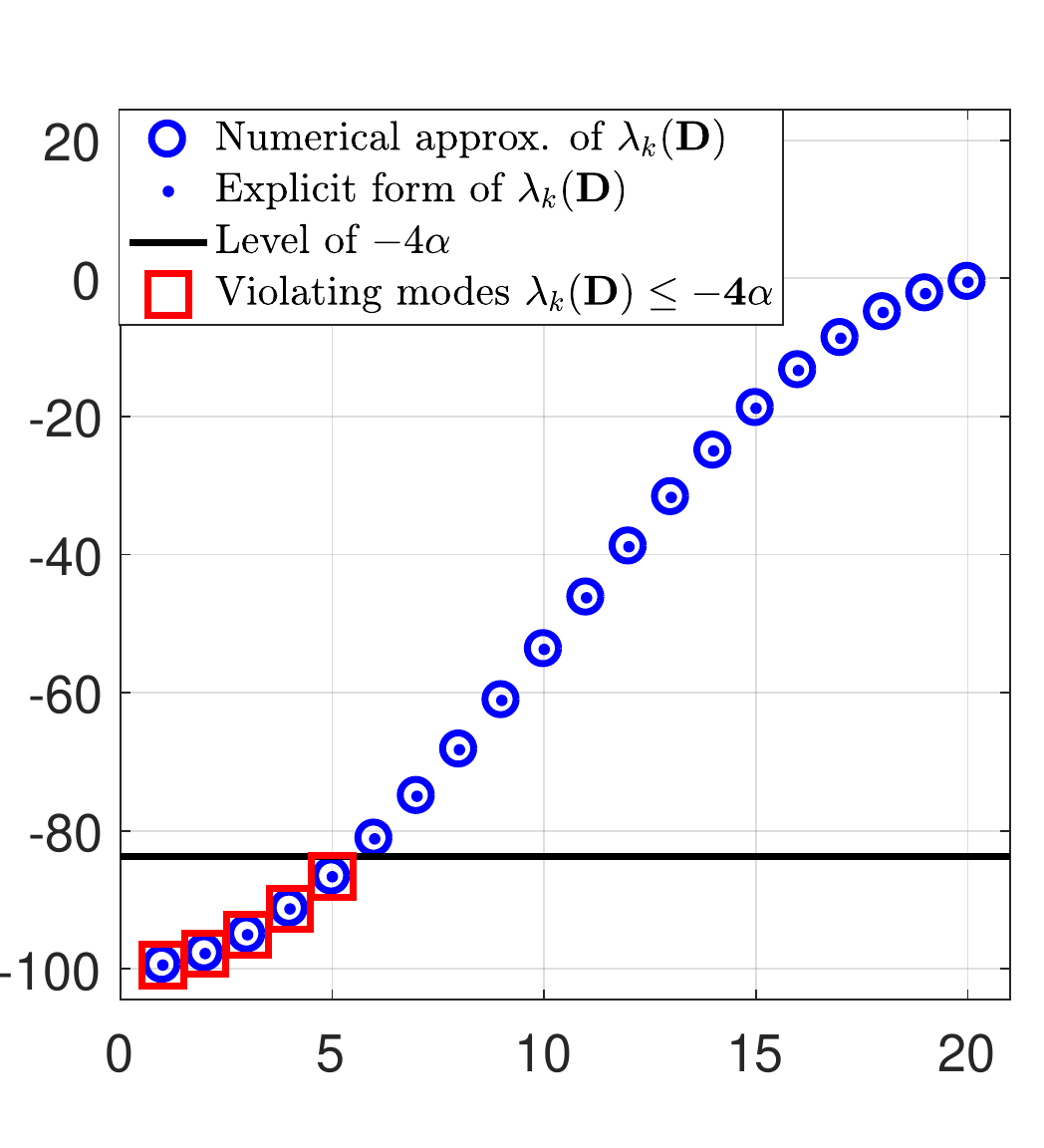}
  \label{fig:eigs_D}
  } \hspace{8mm}
  \subfigure[Eigenvalues of $\m{A}$.]{
  \includegraphics[scale=0.58, trim=150 230 130 230, clip]{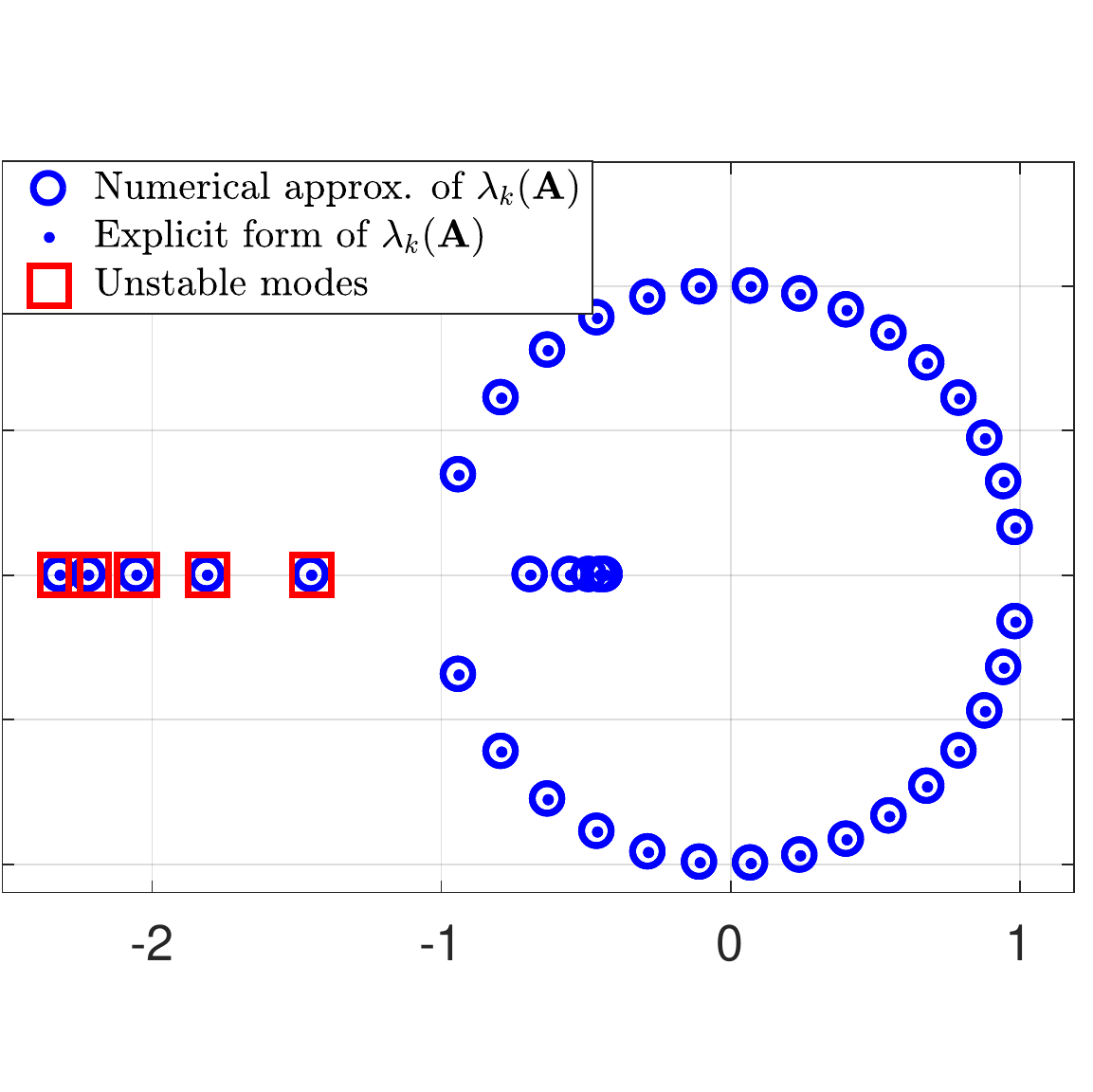}
  \label{fig:eigs_A}
  }
\caption{Demonstrating unstable modes of the discrete SSM.}
\label{fig:stability}
\end{figure}

\begin{theorem}
\label{theorem:observability}
If we choose the parameters $c_0,\tau,\Delta_t,\Delta_z\in\IR_+$ such that the eigenvalues of $\m{A}$ are pairwise distinct and non-zero then the state space model in \eqref{equ:RECENDT_002} is observable. 
\end{theorem}
\begin{proof}
First, we prove the observability by applying the so-called Popov--Belevitch--Hautus (PBH) test. To this end, we give the explicit formulae  for the left and right eigenvectors of $\m{A}$. Lemma~\ref{lemma:eigvecs_A} and \eqref{eq:trans_ssm} implies that $\m{A}= \m{T}^T \m{\widetilde{W}} \m{\Lambda} \m{\widetilde{W}}^{-1} \m{T}=\m{Q}\m{\Lambda}\m{Q}^{-1}$, where $\m{Q}$ is as follows: 
\begin{equation}
	\m{Q}=	\m{T}^T \m{\widetilde{W}} =
	\begin{bmatrix}
    \m{V} & \m{V}\\
    \m{V}\m{\Lambda}_{+}^{-1} & \m{V}\m{\Lambda}_{-}^{-1}
	\end{bmatrix}\,.
\label{eq:right_eigv}
\end{equation}
According to the PBH condition, the SSM in \eqref{equ:RECENDT_002} is observable if and only if $\m{c}^T \m{q}_k \neq 0$ for all $k=1,\ldots,2N_{\m{z}}$, i.e. none of the right eigenvectors $\m{q}_k$ of $\m{A}$ are in the nullspace of $\m{c}^T$. Recalling that $\m{c}^T=\left[1,0,\ldots,0\right]^T$ and using the definition of $\m{v}_k$ in \eqref{eq:eigv_D}, we have  
\begin{equation}
	\m{c}^T \m{q}_k=\m{v}_k\left[1\right]=\sqrt{\frac{2}{N_{\m{z}}+1}} \sin \left( \frac{k\pi}{N_{\m{z}}+1} \right)\qquad (k=1,\ldots,N_{\m{z}})\,.
\label{eq:PHB_obs}
\end{equation}
Note that the dot product $\m{c}^T \m{q}_k$ gives the first coordinate of the eigenvector $\m{v}_k$ of $\m{D}$, which is never equal to zero. Therefore, the PHB condition is satisfied and the SSM is observable.
\end{proof}

\begin{theorem}
\label{theorem:controllability}
If we choose the parameters $c_0,\tau,\Delta_t,\Delta_z\in\IR_+$ such that the eigenvalues of $\m{A}$ are pairwise distinct and non-zero then the state space model in \eqref{equ:RECENDT_002} is controllable if and only if there exists no $s\in\IR_+$ such that 
\begin{equation}
	\ve{d}[j] = s\cdot \sin \left( \frac{j\pi}{N_{\m{z}}+1} \right)\qquad (j=1,\ldots,N_{\m{z}})\,,	
\label{eq:controlab_cond}
\end{equation}
where $\ve{d}$ is generated by the corresponding absorption profile $\bm{\mu}$ in \eqref{equ:RECENDT_014}. 
\end{theorem}
\begin{proof}
In order to analyze the controllability, we consider the matrix $\m{Q}^{-1}$ that contains the left eigenvectors of $\m{A}$ in its rows:  
\begin{equation}
	\m{Q}^{-1}=\m{\widetilde{W}}^{-1} \m{T} = \m{S} \m{P}=
	\begin{bmatrix}
     \left(\m{\Lambda_{+}}-\m{\Lambda_{-}}\right)^{-1} & \m{0}\\
		 \m{0} & \left(\m{\Lambda_{+}}-\m{\Lambda_{-}}\right)^{-1}
	\end{bmatrix}	
	\begin{bmatrix}
    \phantom{-}\m{\Lambda_{+}}\m{V}^T & -\m{\Lambda_{+}}\m{\Lambda_{-}}\m{V}^T\\
    -\m{\Lambda_{-}}\m{V}^T & \phantom{-}\m{\Lambda_{+}}\m{\Lambda_{-}}\m{V}^T
	\end{bmatrix}.	
\label{eq:left_eigv_contr}
\end{equation}
For the sake of simplicity, we will apply the PHB test on the rows of $\m{P}$ only, since multiplying by $\m{S}$ is just a scaling of the eigenvectors. Namely, we should verify that $\m{P}\ve{g}\neq \ve{0}$, where $\ve{g}$ is defined in \eqref{equ:RECENDT_019} and the matrix vector product is the following:
\begin{equation}
	\m{P}\m{g}=
	- \frac{\beta \chi}{C_p \Delta_t}\cdot 
	\begin{bmatrix}
    \phantom{-}\m{\Lambda_{+}}\m{V}^T & -\m{\Lambda_{+}}\m{\Lambda_{-}}\m{V}^T\\
    -\m{\Lambda_{-}}\m{V}^T & \phantom{-}\m{\Lambda_{+}}\m{\Lambda_{-}}\m{V}^T
	\end{bmatrix}
  \begin{bmatrix} 
	\m{V} \m{\Lambda}_1^{-1}\m{V}^T \ve{d} \\ 
	\m{0} 
	\end{bmatrix}=
	- \frac{\beta \chi}{C_p \Delta_t}\cdot
  \begin{bmatrix} 
	\phantom{-}\m{\Lambda_+} \m{\Lambda}_1^{-1}\m{V}^T \ve{d} \\ 
	-\m{\Lambda_-} \m{\Lambda}_1^{-1}\m{V}^T \ve{d}
	\end{bmatrix}
	\,.	
\label{eq:left_eigv}
\end{equation}
Since the diagonal elements of $\m{\Lambda_+},\,\m{\Lambda_-},\,\m{\Lambda}_1$ are non-zero, the $k$th coordinate of the product $\m{P}\m{g}$ will be equal to zero if and only if $\ve{v}_k^T \ve{d}=0$. The orthogonality of $\m{V}$ implies that this condition can be satisfied if and only if $\exists i\neq k$ : $\ve{d}=s \cdot \ve{v}_i\;(s\in\IR\setminus\left\{0\right\})$. Recall that $\ve{v}_i\,(1\leq i \leq N_{\m{z}})$ always has negative coordinates except for $i=1$, when $\ve{v}_1[j]>0\;(1 \leq j \leq N_{\m{z}})$. Due to this fact and to the non-negativity of the absorption profile, $\ve{v}_i^T \ve{d}=0$ if and only if $\ve{d}=s \cdot \ve{v}_1\,(s\in\IR_+)$, which along with \eqref{eq:eigv_D} prove our statement.    
\end{proof}

\begin{corollary}
\label{cor:minrel}
Under the assumptions of Theorem~\ref{theorem:controllability} the state space realization in \eqref{equ:RECENDT_019} is minimal, i.e. the representation is unique up to a similarity transform. 
\end{corollary}

We emphasize that the conditions in Theorems~\ref{theorem:stability}-\ref{theorem:observability} can be easily verified since the eigenvalues are defined exactly in Lemma~\ref{lemma:eig_A}. Therefore, one can set the time resolution $\Delta_t$, the spatial resolution $\Delta_z$, the relaxation time $\tau$, and the ultrasound wave velocity $c_0$ in such a way that the corresponding SSM is asymptotically stable, observable and controllable. Another advantage of the proposed method is that the state matrix can be transformed into a diagonal canonical form by using Lemma~\ref{lemma:eigvecs_A}. Therefore, the transfer function of the SSM can be defined via partial fraction expansion, which permits frequency-domain analysis as well.

% -------------------------------------------------------------------
% Linear Model Formulation
% ----------------7---------------------------------------------------
\section{Linear Model Formulation}
\label{sec:Linear_Model_Formulation}

We will now derive a linear connection between $\ve{d}$ and the measurements. For that, we introduce the vector representation of all measurements $y_k$ for $k = 0, \hdots, N_\ve{y}-1$ as $\ve{y} \in \mathbb{R}^{N_\ve{y} \times 1}$. We begin with the impulse response $h_{u,y}[k]$ from the input $u_k$ to the measurements, which is given as
\begin{align}
h_{u,y}[k] = \begin{cases} 0, & \text{for } k=0 \\
\ve{c}^T \m{A}^{k-1} \ve{g}, & \text{for } 1 \leq k \leq N_\ve{y} -1 \end{cases}, \label{equ:RECENDT_023}
\end{align}
We write the elements of $h_{u,y}[k]$ for $k = 0, \hdots, N_\ve{y}-1$ compactly in vector form as
\begin{align}
\ve{h}_{u,y} = \begin{bmatrix}
0 \\  \ve{c}^T  \\  \ve{c}^T \m{A} \\ \ve{c}^T \m{A}^2 \\ \vdots \\ \ve{c}^T  \m{A}^{N_\ve{y}-2}
\end{bmatrix}  \ve{g}. \label{equ:RECENDT_024}
\end{align}
$u_k$ is connected with the laser intensity according to $u_k = i_k - i_{k-1}$ (cf. \eqref{equ:RECENDT_014}). In other words, the impulse response from $i_k$ to $u_k$, denoted by $h_{i,u}[k]$, has a length of 2 and is given by $h_{i,u}[k] = \left[ 1, \, -1 \right]^T$. The impulse response from the laser intensity $i_k$ to the measurements therefore follows to 
$h_{i,y}[k] = h_{i,u}[k] * h_{u,y}[k]$, or in vector form as
\begin{align}
\ve{h}_{i,y} =& \underbrace{ \left( \begin{bmatrix}
0 \\  \ve{c}^T \\ \ve{c}^T\m{A} \\ \ve{c}^T\m{A}^2 \\ \vdots \\  \ve{c}^T\m{A}^{N_\ve{y}-2}
\end{bmatrix} - \begin{bmatrix}
0 \\ 0 \\  \ve{c}^T \\  \ve{c}^T\m{A} \\ \vdots \\  \ve{c}^T\m{A}^{N_\ve{y}-3}
\end{bmatrix}  \right) }_{\m{M}_6} \ve{g} \label{equ:RECENDT_027} \\
=&\, \m{M}_6 \ve{g}. \label{equ:RECENDT_028}
\end{align}
Note that in contrast to $h_{u,y}[k]$, $h_{i,y}[k]$ practically decreases to zero and can be well approximated as finite impulse response (FIR) with length $N_\ve{y}$. We assume the discretized laser intensity $i_k$ has significant values within the first $N_\ve{i}<N_\ve{y}$ time steps and is zero for all remaining time steps. These $N_\ve{i}$ values of $i_k$ written in vector form are denoted by $\ve{i} \in \mathbb{R}^{N_\ve{i} \times 1}$. In the following, we refer to $\ve{i}$ as laser modulation signal.

The vector of measurements is then given by 
\begin{align}
\ve{y} =\m{C} \ve{h}_{i,y} + \ve{w}, \label{equ:RECENDT_030}
\end{align}
where $\m{C} \in \mathbb{R}^{N_\ve{y} \times N_\ve{y}}$ represents the corresponding discrete convolution operator, i.e. it is a Toeplitz matrix constructed by the elements of $\ve{i}$, and $\ve{w}$ is a zero mean Gaussian noise vector containing the noise samples $w_k$ in \eqref{equ:RECENDT_003}. Combining \eqref{equ:RECENDT_028} and \eqref{equ:RECENDT_030} leads to
\begin{align}
\ve{y} =&\,  \m{C} \m{M}_6 \ve{g} + \ve{w}=\m{C} \m{M}_6 \begin{bmatrix} \m{I} \\ \m{0} \end{bmatrix} \ve{f} + \ve{w}=\m{C} \m{M}_6 \begin{bmatrix} \m{I} \\ \m{0} \end{bmatrix} \m{M}_1^{-1}\ve{b} + \ve{w}. 
\label{equ:RECENDT_031}
\end{align}
%$\ve{g}$ is connected with $\ve{f}$ according to \eqref{equ:RECENDT_019}, yielding
%\begin{align}
%\ve{y} =&\,  \m{C} \m{M}_6 \begin{bmatrix} \m{I} \\ \m{0} \end{bmatrix} \ve{f} + \ve{w}. \label{equ:RECENDT_031b}
%\end{align}
%Inserting the expression for $\ve{f}$ from \eqref{equ:RECENDT_017b} into \eqref{equ:RECENDT_031b} produces
%\begin{align}
%\ve{y} =&\,  \m{C} \m{M}_6 \begin{bmatrix} \m{I} \\ \m{0} \end{bmatrix} \m{M}_1^{-1}\ve{b} + \ve{w}. \label{equ:RECENDT_032}
%\end{align}
Furthermore, $\ve{b}$ is connected with $\ve{d}$ according to \eqref{equ:RECENDT_014}, allowing for
\begin{align}
\ve{y} =&\,  \underbrace{- \frac{\beta \chi}{C_p \Delta_t} \m{C} \m{M}_6\begin{bmatrix} \m{I} \\ \m{0} \end{bmatrix} \m{M}_1^{-1} }_{\m{H}}\ve{d} + \ve{w} \label{equ:RECENDT_033} \\
=& \, \m{H}\ve{d} + \ve{w}. \label{equ:RECENDT_034}
\end{align}
This result shows that the unknown vector $\ve{d}$, which contains the unknown absorption profile according to \eqref{equ:RECENDT_012a}, is linearly connected with the measurement vector $\ve{y}$ in our discretized model.

\section{Estimation of the Absorption Profile}
\label{sec:Estimation}

In this section we apply various estimators on the measurement vector $\ve{y}$ in order to estimate the vector $\ve{d}$. It is followed by a non-linear procedure that allows for estimating the absorption profile $\bm{\mu}$ based on the estimates of $\ve{d}$. By doing so, instead of estimating the unknown pressure profile $\ve{p}_k$, we can directly estimate the vectors $\ve{d}$ and $\bm{\mu}$ based on the measurements $\ve{y}$.

The noise vector $\ve{w}$ in \eqref{equ:RECENDT_034} is assumed to consist of $N_\ve{y}$ white Gaussian noise samples with variance $\sigma_w^2$. Hence, the noise covariance matrix is a scaled identity matrix $\m{C}_{\ve{w}\ve{w}} = \sigma_w^2 \m{I}$. With this condition, the optimal estimator for the model in \eqref{equ:RECENDT_034} in a least squares (LS) sense is given by \cite{Kay-Est.Theory} as follows:
\begin{align}
\hat{\ve{d}} = \left( \m{H}^T \m{H} \right)^{-1} \m{H}^T \ve{y}. \label{equ:RECENDT_036a}
\end{align}
Although this is the best linear unbiased estimator (BLUE) of $\ve{d}$, it is only applicable if $\m{H}$ has full column rank. Typically, the matrix $\m{H}$ is highly ill-conditioned, thus the solution $\hat{\ve{d}}$ is very sensitive to small perturbations, therefore regularization is inevitable. In our experiments we apply two widely used direct methods: the Tikhonov regularization and the truncated/damped singular value decomposition (TSVD/DSVD). These are biased estimators \cite{ridgereg}, however, they can give estimates with less mean squared error (MSE) by choosing a proper regularization parameter. According to Chapter 6.1 in \cite{regmethods}, the regularization methods can be discussed in a unified framework. Namely, let us consider the SVD of $\IR^{N_{\m{y}} \times N_{\m{z}}}\ni\m{H}=\m{U}\m{\Sigma}\m{V}^T$, where $\IR^{N_{\m{y}} \times N_{\m{y}}}\ni\m{U}=\left(\m{u}_1,\ldots,\m{u}_{N_{\m{y}}}\right),\; \IR^{N_{\m{z}} \times N_{\m{z}}}\ni\m{V}=\left(\m{v}_1,\ldots,\m{v}_{N_{\m{z}}}\right)$ are orthonormal matrices, and $\IR^{N_{\m{y}}\times N_{\m{z}}} \ni \m{\Sigma}=\diag\left(\sigma_1,\ldots,\sigma_{N_{\m{z}}}\right)$ with singular values $\sigma_1 \geq \ldots \geq \sigma_{N_{\m{z}}} \geq 0$. Then the regularized solution to \eqref{equ:RECENDT_034} is of the form
\begin{align}
\hat{\ve{d}}_{\text{reg}} = \sum_{j=1}^{N_{\m{z}}} f_j \frac{\m{u}_j^T \m{y}}{\sigma_j} \m{v}_j, \label{equ:regsol_general}
\end{align}
The filter factors $f_j$'s are responsible for controlling the spectral contents of the solution. Generally, they are chosen to eliminate high-frequency components with small $\sigma_j$. In Section~\ref{sec:experiments}, we will apply the following filter factors:
\begin{align}
f_j^{\text{TSVD}}=
	\begin{cases} 1, & \text{for } j=1,\ldots,k_c \\
								0, & \text{for } j=k_c,\ldots,N_\ve{y} 
	\end{cases}, \quad 
f_j^{\text{DSVD}}=\frac{\sigma_j}{\sigma_j+\omega}, \quad	
f_j^{\text{Tikh}}=\frac{\sigma_j^2}{\sigma_j^2+\omega^2},	
	\label{equ:filt_facts}
\end{align}
where $k_c$ is called truncation or cutoff index, and $\omega>0$ is the regularization parameter. Note that the TSVD applies an ideal filter to the spectral content of the solution, while Tikhonov regularization and DSVD allows smoothing in a wider transition band. In order to control the rate of smoothing the regularization parameters $k_c$ and $\omega$ should be chosen properly. Here, we utilize the work of Hansen and O'Leary \cite{lcurve}, in which they estimate the optimal regularization parameter based on the so-called L-curve. This is a log-log plot of the norm of the regularized solution $\left\|\hat{\ve{d}}_{\text{reg}}\right\|_2$ versus the corresponding reconstruction error $\left\|\m{y}-\m{H} \hat{\ve{d}}_{\text{reg}} \right\|_2$. The point with maximum curvature is called corner, which separates the solutions into under- and overregularized sets. Therefore, choosing $k_c$ and $\omega$ corresponding to the corner point is a good tradeoff between regularization and perturbation errors. Generalized cross-validation (GCV) is another option for estimating the optimal regularization parameter. However, in \cite{lcurve}, it was shown that the corner point of the L-curve is a more robust estimator, especially for highly correlated errors, i.e. when $\m{C}_{\ve{w}\ve{w}}$ is not diagonal. 

Now, in order to estimate the absorption profile, let us recall the relation between $\ve{d}$ and $\bm{\mu}$:
%We now introduce a forward elimination procedure that allows deriving estimates of the absorption profile based on $\hat{\ve{d}}$. $\ve{d}$ and the absorption profile $\bm{\mu}$ are connected via
\begin{align}
\ve{d} = \begin{bmatrix}
d_0 \\ d_1 \\ d_2  \\ \vdots \\  d_{N_\ve{z}-1} \end{bmatrix} = \begin{bmatrix}
\mu_0 \\ \mu_1 a_0 \\ \mu_2 a_1 a_0 \\ \vdots \\  \mu_{N_\ve{z}-1} a_{N_\ve{z}-2} a_{N_\ve{z}-3} \hdots a_1 a_0
\end{bmatrix} \label{equ:RECENDT_035}
\end{align}
(cf. \eqref{equ:RECENDT_012a}), where $a_n$ is given by $a_n = \mathrm{e}^{-\mu_n \Delta_z}$.
According to \eqref{equ:RECENDT_035}, an estimate of the first element of the absorption profile $\mu_0$ is given by the first element of $\hat{\ve{d}}$ via $\hat{\mu}_0 = \hat{d}_0$. This estimate is utilized to derive an estimate for $a_0$ as $\hat{a}_0 = \mathrm{e}^{-\hat{\mu}_0 \Delta_z}$. Using this result to approximate $a_0$ in \eqref{equ:RECENDT_035} immediately leads to an estimate of the next entry of $\bm{\mu}$ according to $\hat{\mu}_1 = \frac{\hat{d}_1}{\hat{a}_0}$. Generally, the coordinates of $\hat{\bm{\mu}}$ can be computed as follows:
\begin{equation}
	\hat{\mu}_{n}\,= \frac{\hat{d}_{n}}{\prod_{j=0}^{n-1}\hat{a}_{j}}=\frac{\hat{d}_{n}}{\exp\left({-\sum_{j=0}^{n-1}\hat{\mu}_{j}}\Delta_z\right)} \qquad (n = 0,\hdots, N_\ve{z}-1)\,.
	\label{eq:mu_recursion}
\end{equation}
%Note that the estimation for $\bm{\mu}$ results in an error propagation from $\hat{\ve{d}}$ to $\hat{\bm{\mu}}$, which can be observed in the simulation examples in Section~\ref{sec:experiments}.

We emphasize that the proposed reconstruction method is very general and widely applicable. It can deal with completely arbitrary laser modulation signals and it is not constrained to signals with a well-behaving autocorrelation function, chirped signals, or signals with varying spaces between short pulses. Many effects are accounted for such as frequency dependent attenuation and a decrease in laser intensity due to absorption inside the probe.

% -------------------------------------------------------------------
% Optimization
% -------------------------------------------------------------------
\section{Optimization of the Laser Intensity Function}
\label{sec:Optimization}
Estimating the absorption profile is a difficult problem due to the ill-conditioned system of linear equations in \eqref{equ:RECENDT_034}, and therefore, we used various methods for estimating the absorption profile. Although the BLUE was the simplest estimator among them, it provides a possibility for finding an optimal laser modulation function. To this end, let us consider the vector $\hat{\bm{\mu}}$ that results from a non-linear transformation of the estimates $\hat{\ve{d}}$. However, the more accurate the estimated vector $\hat{\ve{d}}$ is, the more precise $\hat{\bm{\mu}}$ becomes. Hence, we focus on the estimation accuracy of $\hat{\ve{d}}$ since the measurements are linear in $\ve{d}$ and the error statistics of $\hat{\ve{d}}$ are analytically tractable in case of the BLUE. We investigate if an optimal laser modulation function in terms of the discrete-time laser modulation signal $\ve{i}$ can be found to minimize the error variances in $\hat{\ve{d}}$.

The corresponding error covariance matrix of the estimates $\hat{\ve{d}}$ in \eqref{equ:RECENDT_036a} is given by \cite{Kay-Est.Theory}
\begin{align}
\m{C}_{\hat{\ve{d}}\hat{\ve{d}}} = \sigma_w^2 \left( \m{H}^T \m{H}\right)^{-1}. \label{equ:RECENDT_037}
\end{align}
Note that this error covariance matrix does not account for the process noise in \eqref{equ:RECENDT_002}. With this restriction, it describes the confidence in the estimates $\hat{\ve{d}}$. Furthermore, since $\m{H}$ in  \eqref{equ:RECENDT_036a} is a function of $\ve{i}$ according to \eqref{equ:RECENDT_033}, the confidence in the estimates $\hat{\ve{d}}$ is also a function of $\ve{i}$. This fact is utilized in the following where \eqref{equ:RECENDT_037} is used as basis for formulating an optimization problem in order to find an optimal laser modulation signal $\ve{i}$.

The main diagonal elements of $\m{C}_{\hat{\ve{d}}\hat{\ve{d}}}$ in \eqref{equ:RECENDT_037} contain the variances of the estimates. In order to minimize these variances we use the trace of $\m{C}_{\hat{\ve{d}}\hat{\ve{d}}}$ as cost function
\begin{align}
J(\ve{i}) = \mathrm{trace} \left(\m{C}_{\hat{\ve{d}}\hat{\ve{d}}}\right). \label{equ:RECENDT_038}
\end{align}
Now the goal is to find the vector $\ve{i}$ that minimizes this cost function. However, we define several constraints on $\ve{i}$, such that it represents a practicable modulation sequence. Some obvious constraints are as follows:
\begin{enumerate}
\item $\ve{i}$ must be time-limited. 
\item The elements of $\ve{i}$ must be larger than or equal to $0$ and smaller than or equal to $1$. The first requirement origins from the fact that a laser intensity is always positive from a physical point of view. The upper bound accounts for the restriction that every laser has a maximum output power that cannot be exceeded. 
\item The energy shall be below some upper limit.
\item $\ve{i}$ shall be band-limited in order to easily transform $\ve{i}$ into a continuous-time laser modulation function $i(t)$ without aliases.
\end{enumerate}
These are just some meaningful examples of possible constraints. In practice, maybe more constraints origin from the concrete laser setup. In Sections~\ref{sec:Demonstration_1}-\ref{sec:Demonstration_4}, we present two examples applying the same constraints as listed above, and discuss the resulting optimized laser modulation signals. Note that the cost function in \eqref{equ:RECENDT_038} as well as some constraints are non-linear in $\ve{i}$. Hence, the optimization problem is very demanding and a global minimum is unlikely to be found. However, local minima can be found by numerical optimization. We will show the potential in this approach by presenting several simulation results in the next section. We begin with two examples where our reconstruction method is demonstrated. After that, we discuss the performance gain of the proposed optimization procedure.

\section{Simulations}
\label{sec:experiments}
In these demonstrations, we estimate the absorption profile from synthesized measurement data $\m{y}$ with the methods described in Section~\ref{sec:Estimation}. This gives an impression about the accuracy of the proposed algorithm for estimating $\ve{d}$ and the absorption profile $\bm{\mu}$. The following numerical experiments and the corresponding MatLab implementations can be found as supplementary material of the paper, in which we used external libraries such as the regularization tool package \cite{regtool}. It is also worth mentioning that, for real practical problems, positivity of the absorption profile $\bm{\mu}$, and thus the positivity of $\m{d}$ is a valid assumption. Therefore, we included the non-negative Tikhonov regularization in our comparative study, that we implemented by using the \texttt{lsqnonneg} MatLab routine. 

Note that in the following, we do not compare our reconstruction method with a competing approach. The reason for this is that, to the best of our knowledge, there exists no reconstruction method in literature that allows for arbitrary absorption profiles, arbitrary laser modulation signals, that accounts for both the absorption of the laser intensity with increasing depth, and the frequency dependent attenuation of the ultrasound waves.
  
% -------------------------------------------------------------------
% Demonstration 1
% -------------------------------------------------------------------
\subsection{Example for estimating a smooth absorption profile}
\label{sec:Demonstration_1}

We begin with a single reconstruction task. For that, we choose the parameters listed in Tab.~\ref{tab:Example1}. Furthermore, as laser modulation signal $i_k$ we take a short pulse with a duration of approximately $10\,ns$ as shown in Fig.~\ref{fig:Example1fff}. 

\begin{table}[ht]
\begin{center}
	\caption{Model parameters used for the first demonstration example. \label{tab:Example1} }
  \begin{tabular}{ | c | c | c | }
  	\hline
  	\textbf{Parameter} & \textbf{Value} & \textbf{Unit} \\
    \hline \hline
    $\Delta_t$ 		& $10^{-9}$ 			& s 	\\ \hline
    $\Delta_z$ 		& $30 \cdot 10^{-7}$ 	& m 	\\ \hline
%    $N_\ve{z}$ 		& $100$ 				& 		\\ \hline
%    $n_l$ 			& $0$ 					& 		\\ \hline
%    $n_r$ 			& $20$ 					& 		\\ \hline
    $N_\ve{z}$ 		& $20$ 					& 		\\ \hline
    $N_\ve{y}$ 		& $100$ 				& 		\\ \hline
    $N_\ve{i}$ 		& $20$ 					& 		\\ \hline
    $\sigma_q^2$ 	& $10^{-28}$  			& $N^2/m^4$ 		\\ \hline
    $\sigma_w^2$ 	& $10^{-10}$  			& $N^2/m^4$		\\ \hline
    $\tau$ 			& $77 \cdot 10^{-12}$  	& s		\\ \hline
    $\chi$ 			& $3\cdot 10^{-2}$  	& 		\\ \hline
    %$C_p$ 			& $1$  					& 		\\ \hline
    $\beta / C_p$ 		& $1 / 1$  					& 		\\ \hline
    $c_0$ 			& $1500$  				& m/s	\\ \hline
  %  $\alpha$ 		& $0$ 					& 		\\ \hline
  \end{tabular}
\end{center}
\end{table}

The measurement vector $\ve{y}$ is generated using the state space model in \eqref{equ:RECENDT_002} and \eqref{equ:RECENDT_003}, which is indicated by the black curve in Fig.~\ref{fig:Example1fff}. Based on the measurements, $\ve{d}$ is estimated according to \eqref{equ:RECENDT_034} by applying the methods presented in Section~\ref{sec:Estimation}. Then an estimate of the true absorption profile is obtained via \eqref{eq:mu_recursion}, which is displayed in Fig.~\ref{fig:Example1}. Here, the measurement noise $\sigma_w^2$ is equal to $10^{-10}$, which results in a measurement vector $\m{y}$ with signal-to-noise ratio $\text{SNR}(\m{y})=93.6$ dB. In this setting, the estimated $\hat{\bm{\mu}}$ is very close to the true $\bm{\mu}$ and the estimation preserves the main characteristic of the original absorption profile.

Repeatedly performing this estimation procedure allows deriving statistics about the estimation error. The resulting root mean square errors $\text{RMSE}(\hat{\ve{d}})= \sqrt{E[ \| \ve{d} - \hat{\ve{d}} \|_2^2]}$ and $\text{RMSE}(\hat{\bm{\mu}})= \sqrt{E[ \| \bm{\mu} - \hat{\bm{\mu}} \|_2^2]}$ plotted over the spatial coordinate are shown in Fig.~\ref{fig:Example1_sdt_dev}. This graph reveals that the errors tend to increase as the spatial depth increases. This behavior has several reasons. 1) The arriving laser intensity decreases due to absorption inside the probe. 2) The damping of the ultrasound waves between generation and arriving at the sensor increases as the spatial depth increases. 3) Since the ultrasound waves generated at larger indexes $n$ require a longer time to arrive at the sensor, more noise samples $q_k$ in \eqref{equ:RECENDT_002} influence the measurements. The fact that both, the error of $\hat{\ve{d}}$ and the error of $\hat{\bm{\mu}}$ decrease for indexes $n \geq 17$ is due to the decreasing values of $\ve{d}$ and $\bm{\mu}$ for these indexes and strongly depends on the absorption profile. We also draw the attention to the fact that the error of $\hat{\ve{d}}$ is smaller or equal to the error of $\hat{\bm{\mu}}$ (see e.g. Fig.~\ref{fig:Example1_sdt_dev}). This phenomenon is due to the successive estimation of the absorption profile in \eqref{eq:mu_recursion}. Therefore, the errors of $\hat{\mu}_n$ decrease the accuracy of all estimates with larger indices. 

In order to investigate the impact of the measurement noise on the accuracy of our method, we varied $\sigma^2_w$ between $10^{-10}$ and $10^{-1}$. One can see the results in Fig.~\ref{fig:example1_opt}, where we displayed the averaged (A)RMSE of $100$ simulation runs for each value of $\sigma^2_w$ with the corresponding SNR's indicated on the top of the plot. For both quantities $\hat{\ve{d}}$ and $\hat{\bm{\mu}}$, the BLUE shows the worst estimation accuracy as the variance of the noise increases, while the non-negative Tikhonov regularization achieves the best performance with the lowest ARMSE. The latter is not surprising, since we are utilizing an additional information during the estimation, which is the non-negativity of the absorption profile. In Fig.~\ref{fig:example1_reconstruction_with_high_noise}, we are presenting an example as a worst-case scenario. Although the SNR is very low, i.e. $3.5$ dB, the main characteristic of the absorption profile is retained by the Tikhonov and the non-negative Tikhonov regularizations. Namely, the two main side peaks and the small middle peak can be observed in those approximations. However, the damped SVD gives a highly oscillating estimation. In this case, the BLUE estimation was so bad we excluded it from the plot.

We emphasize that basically the same model utilized for generating the measurements was also used by the reconstruction method except for the process noise in \eqref{equ:RECENDT_002}. However, to the best of our knowledge, there exists no alternative simulator that employs Stokes' PDE and that allows for arbitrary laser modulation signals and arbitrary absorption profiles including all mentioned effects that we account for in this work. Furthermore, we highlight that the process noise in \eqref{equ:RECENDT_002} results in a random mismatch between the model used for generating the measurements and the model used by the reconstruction algorithm. 

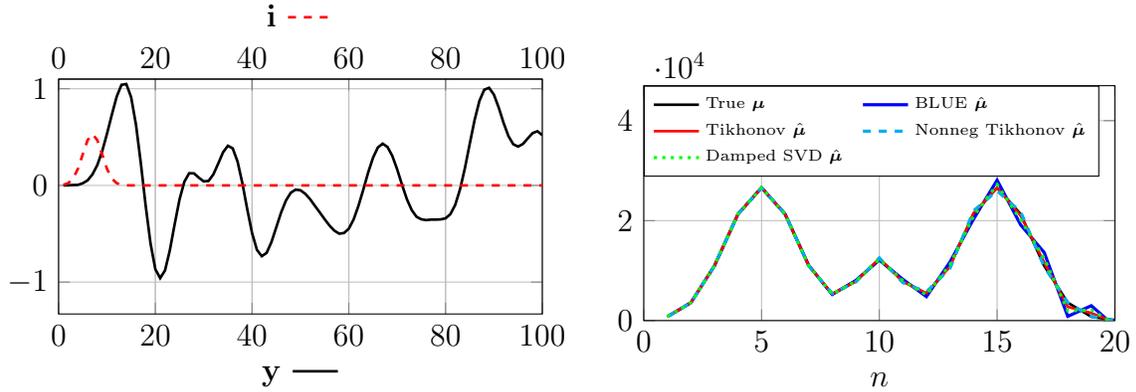
\begin{figure}[!t]
\centering
\vspace{2mm}
  \subfigure[The laser modulation signal $\ve{i}$ and the measured pressure signal $\ve{y}$.]{
	\label{fig:Example1fff}
\begin{tikzpicture}
\pgfplotsset{set layers}
\begin{axis}[
width=1.9\columnwidth, height = .3\columnwidth,
xmin=0,xmax=100,
ymin=-1.33,ymax=1.1,
xlabel={$\ve{y}$ \ref{fig:Example1fff_axis1} },
unit vector ratio*=1 20 1,
        legend style={at={(1,0)}, anchor=south east,
								/tikz/column 2/.style={
                column sep=5pt,
            },
        font=\tiny},
grid=major,
]
%\addplot[line width=1pt, color=black] table[x index =0, y index =1] {./Single_estimation_result_y_ex1.dat};
\addplot[line width=1pt, color=black] table[x index =0, y index =1] {./Example_1_y_and_low_noise.dat};
\label{fig:Example1fff_axis1} 
\end{axis}
\pgfplotsset{set layers}
\begin{axis}[
width=1.9\columnwidth, height = .3\columnwidth,
xmin=0,xmax=100,
ymin=-1.33,ymax=1.1,
axis y line=none,
axis x line*=top,
xlabel style = {align=center},
xlabel={$\ve{i}$ \ref{pgfplots:pulse} },
unit vector ratio*=1 20 1,
        legend style={at={(1,0)}, anchor=south east,
								/tikz/column 2/.style={
                column sep=5pt,
            },
        font=\tiny},
grid=major,
]
%\addplot[line width=1pt, color=black] table[x index =0, y index =1] {./Single_estimation_result_y_ex1.dat};
\addplot[line width=1pt, color=red, style=dashed] table[x index =0, y index =3] {./Example_1_y_and_low_noise.dat};
\label{pgfplots:pulse}
\end{axis}
\end{tikzpicture}%  \label{fig:Example1}
  } \hspace{0mm}
	\subfigure[The true and the estimated absorption profile.]{
		\label{fig:Example1}	
	\begin{tikzpicture}
    \begin{axis}[%
    compat=newest, 
		width=0.5\columnwidth, height = .3\columnwidth,
		ylabel style={align=center}, 
		ylabel style={text width=3.4cm}, 
		xlabel=$n$,
		legend pos=north west, 
		legend cell align=left,
		legend columns=2, 
        legend style={at={(0,1)}, anchor=north west,
								/tikz/column 2/.style={
                column sep=5pt,
            },
        font=\tiny},
		xmin = 0,
		xmax = 20,
		ymax = 4.7e4,
		ymin = 0,
		grid=major]
			\addplot[line width=1pt, color=black] table[x index =0, y index =1] {./Example_1_low_noise_mu.dat};
			\addlegendentry{True $\bm{\mu}$}
			\addplot[line width=1.2pt, color=blue] table[x index =0, y index =2] {./Example_1_low_noise_mu.dat};
			\addlegendentry{BLUE $\hat{\bm{\mu}}$}
			\addplot[line width=1pt, color=red] table[x index =0, y index =3] {./Example_1_low_noise_mu.dat};
			\addlegendentry{Tikhonov $\hat{\bm{\mu}}$}
			\addplot[line width=1.2pt, color=cyan, style=dashed] table[x index =0, y index =4] {./Example_1_low_noise_mu.dat};
			\addlegendentry{Nonneg Tikhonov $\hat{\bm{\mu}}$}
			\addplot[line width=1.2pt, color=green, style=dotted] table[x index =0, y index =5] {./Example_1_low_noise_mu.dat};
			\addlegendentry{Damped SVD $\hat{\bm{\mu}}$}		
    \end{axis}	
	\end{tikzpicture}	
  } \hspace{0mm}
	\centering
\caption{Application of the linear model for estimating the absorption profile.}
\label{fig:example1_reconstruction}
\end{figure}

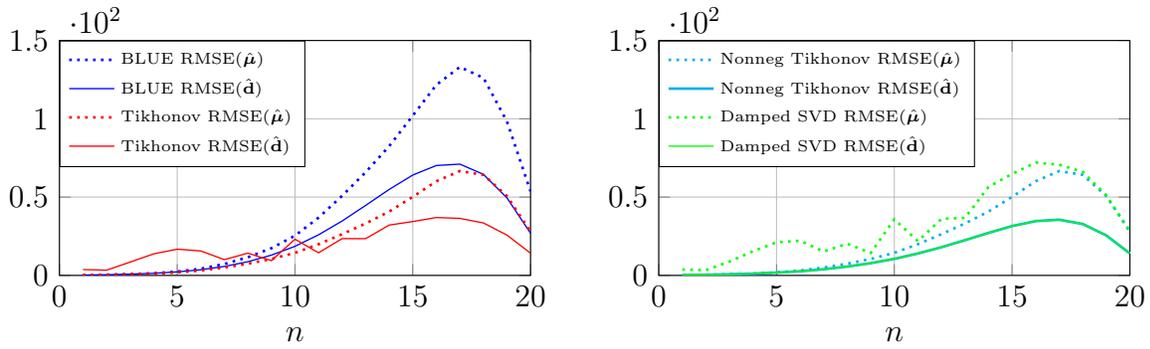
\begin{figure}[!t]
\centering
\vspace{2mm}
  \subfigure[Relation between the RMSE of $\hat{\ve{d}}$ and $\hat{\bm{\mu}}$ for the BLUE and the Tikhonov regularization.]{
\begin{tikzpicture}

    \begin{axis}[%
		  scaled y ticks=base 10:-2,
    	name=plot1,
    	compat=newest, 
		width=0.5\columnwidth, height = .3\columnwidth, xlabel=$n$, 
		ylabel style={align=center}, 
		ylabel style={text width=3.4cm},
		%ylabel={$y_k$}, 
		%scaled y ticks=base 10:-3,
		legend pos=north west, 
		legend cell align=left,
		legend columns=1, 
        legend style={at={(0,1)}, anchor=north west,
                    % the /tikz/ prefix is necessary here...
                    % otherwise, it might end-up with `/pgfplots/column 2`
                    % which is not what we want. compare pgfmanual.pdf
            /tikz/column 2/.style={
                column sep=5pt,
            },
        font=\tiny},
		xmin = 0,
		xmax = 20,
		ymax = 1.5e2,
		ymin = 0,
		grid=major]
			\addplot[line width=1pt, color=blue, style=dotted] table[x index =0, y index =1] {./Example_1_ARMSE_low_noise.dat};
			\addlegendentry{{BLUE $\text{RMSE}(\hat{\bm{\mu}})$}}
			
			\addplot[line width=0.5pt, color=blue] table[x index =0, y index =2] {./Example_1_ARMSE_low_noise.dat};
			\addlegendentry{{BLUE $\text{RMSE}(\hat{\ve{d}})$}}
			
			\addplot[line width=1pt, color=red, style=dotted] table[x index =0, y index =3] {./Example_1_ARMSE_low_noise.dat};
			\addlegendentry{{Tikhonov $\text{RMSE}(\hat{\bm{\mu}})$}}
			
			\addplot[line width=0.5pt, color=red] table[x index =0, y index =4] {./Example_1_ARMSE_low_noise.dat};
			\addlegendentry{{Tikhonov $\text{RMSE}(\hat{\ve{d}})$}}
						
    \end{axis}
    \label{fig:Example1_sdt_dev_a}
\end{tikzpicture}
  } \hspace{0mm}
	\subfigure[Relation between the RMSE of $\hat{\ve{d}}$ and $\hat{\bm{\mu}}$ for the non-negative Tikhonov regularization and the damped SVD.]{
\begin{tikzpicture}

    \begin{axis}[%
		  scaled y ticks=base 10:-2,
    	name=plot1,
    	compat=newest, 
		width=0.5\columnwidth, height = .3\columnwidth, xlabel=$n$, 
		ylabel style={align=center}, 
		ylabel style={text width=3.4cm},
		%ylabel={$y_k$}, 
		%scaled y ticks=base 10:-3,
		legend pos=north west, 
		legend cell align=left,
		legend columns=1, 
        legend style={at={(0,1)}, anchor=north west,
                    % the /tikz/ prefix is necessary here...
                    % otherwise, it might end-up with `/pgfplots/column 2`
                    % which is not what we want. compare pgfmanual.pdf
            /tikz/column 2/.style={
                column sep=5pt,
            },
        font=\tiny},
		xmin = 0,
		xmax = 20,
		ymax = 1.5e2,
		ymin = 0,
		grid=major]
			\addplot[line width=1pt, color=cyan, style=dotted] table[x index =0, y index =5] {./Example_1_ARMSE_low_noise.dat};
			\addlegendentry{{Nonneg Tikhonov $\text{RMSE}(\hat{\bm{\mu}})$}}
			
			\addplot[line width=1pt, color=cyan] table[x index =0, y index =6] {./Example_1_ARMSE_low_noise.dat};
			\addlegendentry{{Nonneg Tikhonov $\text{RMSE}(\hat{\ve{d}})$}}
			
			\addplot[line width=1pt, color=green, style=dotted] table[x index =0, y index =7] {./Example_1_ARMSE_low_noise.dat};
			\addlegendentry{{Damped SVD $\text{RMSE}(\hat{\bm{\mu}})$}}
			
			\addplot[line width=0.5pt, color=green] table[x index =0, y index =8] {./Example_1_ARMSE_low_noise.dat};
			\addlegendentry{{Damped SVD $\text{RMSE}(\hat{\ve{d}})$}}
						
    \end{axis}
    \label{fig:Example1_sdt_dev_b}
\end{tikzpicture}
  } \hspace{0mm}
	\centering
\caption{RMSE of $\hat{\ve{d}}$ and $\hat{\bm{\mu}}$  averaged over $100$ simulation runs for all grid points within the area of interest $n = 0, \hdots, N_\ve{z}-1$ with average $\text{SNR}(\m{y})=93.6$ dB.}
\label{fig:Example1_sdt_dev}
\end{figure}

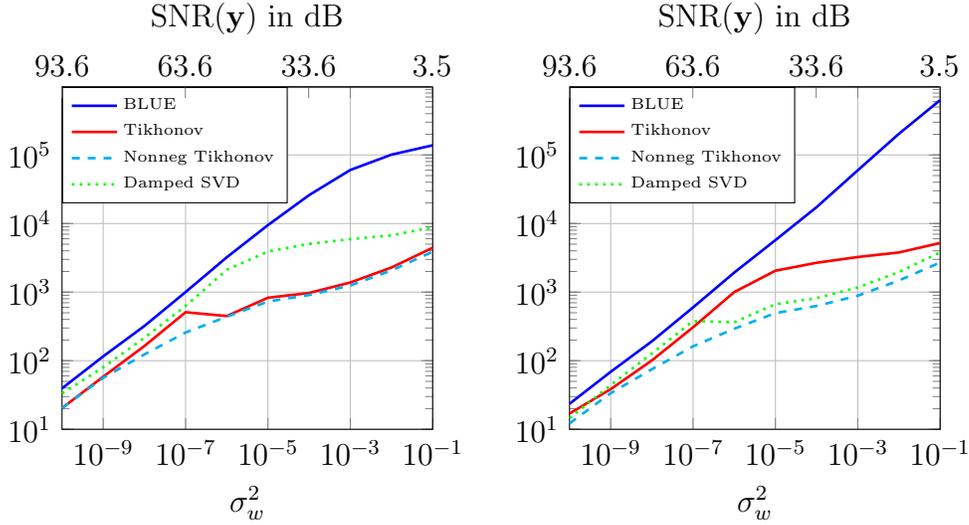
\begin{figure}[!t]
\centering
\subfigure[ARMSE of $\hat{\m{d}}$.]{
	\begin{tikzpicture}
    \begin{loglogaxis}[%
   	name=plot1,
   	compat=newest, 
		unit vector ratio*=0.6 1 1,
		scale=0.8,
		xlabel=$\sigma_w^2$, 
		ylabel style={align=center}, 
		ylabel style={text width=3.4cm},
		scaled x ticks=base 10:-1,
		scaled y ticks=base 10:6,
		legend pos=north west, 
		legend cell align=left,
		legend columns=1, 
		xmin=1e-10, xmax=1e-1,
	  ymin=1e+1, ymax=0.98e+6,
        legend style={at={(0,1)}, anchor=north west,
                    % the /tikz/ prefix is necessary here...
                    % otherwise, it might end-up with `/pgfplots/column 2`
                    % which is not what we want. compare pgfmanual.pdf
            /tikz/column 2/.style={
                column sep=5pt,
            },
        font=\tiny},
		grid=major]
			\addplot[line width=1pt, color=blue] table[x index =0, y index =1] {./Example_1_ARMSE_mu.dat};
			\addlegendentry{{BLUE}}
			
			\addplot[line width=1pt, color=red] table[x index =0, y index =2] {./Example_1_ARMSE_mu.dat};
			\addlegendentry{{Tikhonov}}

			\addplot[line width=1pt, color=cyan, style=dashed] table[x index =0, y index =3] {./Example_1_ARMSE_mu.dat};
			\addlegendentry{{Nonneg Tikhonov}}

			\addplot[line width=1pt, color=green, style=dotted] table[x index =0, y index =4] {./Example_1_ARMSE_mu.dat};
			\addlegendentry{{Damped SVD}}

    \end{loglogaxis}
    \begin{loglogaxis}[%
   	name=plot1,
   	compat=newest, 
		unit vector ratio*=0.6 1 1,
		scale=0.8,
		axis y line=none,
		axis x line*=top,
		xlabel style={align=center}, 
		xlabel style={text width=3.4cm},
		xlabel={$\text{SNR}(\m{y})$ in dB}, 
		xtick=data,
		xticklabels={93.6, 63.6, 33.6, 3.5},
%		scaled x ticks=base 10:-1,
		legend pos=north west, 
		legend cell align=left,
		legend columns=1, 
		xmin=1e-10, xmax=1e-1,
	  ymin=1e+1, ymax=1e+6,
		]
    \end{loglogaxis}	
\end{tikzpicture}
\label{fig:Example1_mu_armse}
}
\subfigure[ARMSE of $\hat{\bm{\mu}}$.]{
	\begin{tikzpicture}
    \begin{loglogaxis}[%
   	name=plot1,
   	compat=newest, 
		unit vector ratio*=0.6 1 1,
		scale=0.8,
		xlabel=$\sigma_w^2$, 
		ylabel style={align=center}, 
		ylabel style={text width=3.4cm},
		scaled x ticks=base 10:-1,
		scaled y ticks=base 10:6,
		legend pos=north west, 
		legend cell align=left,
		legend columns=1, 
		xmin=1e-10, xmax=1e-1,
	  ymin=1e+1, ymax=0.98e+6,
        legend style={at={(0,1)}, anchor=north west,
                    % the /tikz/ prefix is necessary here...
                    % otherwise, it might end-up with `/pgfplots/column 2`
                    % which is not what we want. compare pgfmanual.pdf
            /tikz/column 2/.style={
                column sep=5pt,
            },
        font=\tiny},
		grid=major]
			\addplot[line width=1pt, color=blue] table[x index =0, y index =1] {./Example_1_ARMSE_d.dat};
			\addlegendentry{{BLUE}}
			
			\addplot[line width=1pt, color=red] table[x index =0, y index =2] {./Example_1_ARMSE_d.dat};
			\addlegendentry{{Tikhonov}}

			\addplot[line width=1pt, color=cyan, style=dashed] table[x index =0, y index =3] {./Example_1_ARMSE_d.dat};
			\addlegendentry{{Nonneg Tikhonov}}

			\addplot[line width=1pt, color=green, style=dotted] table[x index =0, y index =4] {./Example_1_ARMSE_d.dat};
			\addlegendentry{{Damped SVD}}

    \end{loglogaxis}
    \begin{loglogaxis}[%
   	name=plot1,
   	compat=newest, 
		unit vector ratio*=0.6 1 1,
		scale=0.8,
		axis y line=none,
		axis x line*=top,
		xlabel style={align=center}, 
		xlabel style={text width=3.4cm},
		xlabel={$\text{SNR}(\m{y})$ in dB}, 
		xtick=data,
		xticklabels={93.6, 63.6, 33.6, 3.5},
%		scaled x ticks=base 10:-1,
		legend pos=north west, 
		legend cell align=left,
		legend columns=1, 
		xmin=1e-10, xmax=1e-1,
	  ymin=1e+1, ymax=1e+6,
		]
    \end{loglogaxis}	
\end{tikzpicture}
\label{fig:Example1_d_armse}
}
\caption{ARMSE of the estimated quantities $\m{d}$ and $\bm{\mu}$ for Example 1.}%
\label{fig:example1_opt}%
\end{figure}

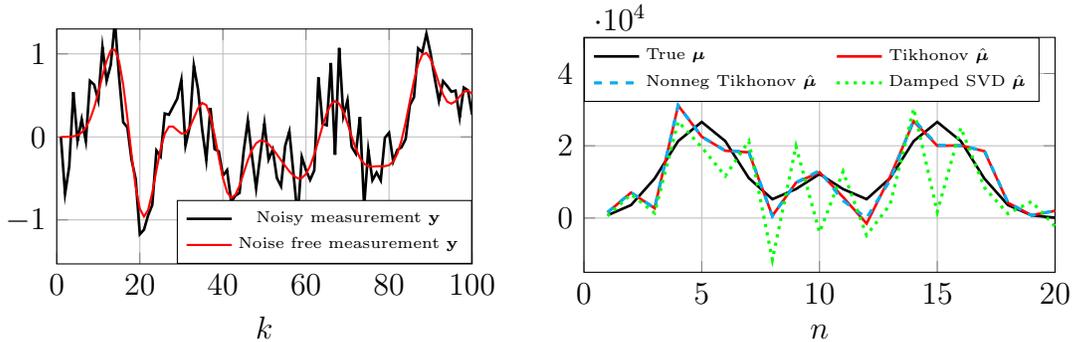
\begin{figure}[!t]
\centering
\vspace{4mm}
	\subfigure[The measurement vector with and without noise.]{
\begin{tikzpicture}
\pgfplotsset{set layers}
\begin{axis}[
width=1.9\columnwidth, height = .3\columnwidth,
xmin=0,xmax=100,
ymin=-1.53,ymax=1.3,
xlabel=$k$,
unit vector ratio*=1 20 1,
        legend style={at={(1,0)}, anchor=south east,
								/tikz/column 2/.style={
                column sep=5pt,
            },
        font=\tiny},
grid=major,
]
%\addplot[line width=1pt, color=black] table[x index =0, y index =1] {./Single_estimation_result_y_ex1.dat};
\addplot[line width=1pt, color=black] table[x index =0, y index =1] {./Example_1_y_and_noise.dat};
\addlegendentry{{Noisy measurement $\m{y}$}}
\addplot[line width=0.8pt, color=red] table[x index =0, y index =2] {./Example_1_y_and_noise.dat};
\addlegendentry{{Noise free measurement $\m{y}$}}
\label{fig:Example1_mu_est_2}
\end{axis}
\end{tikzpicture}
  } \hspace{0mm}
	\subfigure[The true and estimated absorption profiles.]{

\begin{tikzpicture}
    \begin{axis}[%
    compat=newest, 
		width=0.5\columnwidth, height = .3\columnwidth,
		ylabel style={align=center}, 
		ylabel style={text width=3.4cm}, 
		xlabel=$n$,
		legend pos=north west, 
		legend cell align=left,
		legend columns=2, 
        legend style={at={(0,1)}, anchor=north west,
								/tikz/column 2/.style={
                column sep=5pt,
            },
        font=\tiny},
		xmin = 0,
		xmax = 20,
		ymax = 5e4,
		ymin = -15000,
		grid=major]
			\addplot[line width=1pt, color=black] table[x index =0, y index =1] {./Example_1_high_noise_mu.dat};
			\addlegendentry{True $\bm{\mu}$}
%			\addplot[line width=1.2pt, color=blue] table[x index =0, y index =2] {./Example_1_high_noise_mu.dat};
%			\addlegendentry{BLUE $\hat{\bm{\mu}}$}
			\addplot[line width=1pt, color=red] table[x index =0, y index =3] {./Example_1_high_noise_mu.dat};
			\addlegendentry{Tikhonov $\hat{\bm{\mu}}$}
			\addplot[line width=1.2pt, color=cyan, style=dashed] table[x index =0, y index =4] {./Example_1_high_noise_mu.dat};
			\addlegendentry{Nonneg Tikhonov $\hat{\bm{\mu}}$}
			\addplot[line width=1.2pt, color=green, style=dotted] table[x index =0, y index =5] {./Example_1_high_noise_mu.dat};
			\addlegendentry{Damped SVD $\hat{\bm{\mu}}$}
	  \label{fig:Example1_mu_est_3}			
    \end{axis}
	\end{tikzpicture}
  }
\caption{Application of the linear model for estimating the absorption profile with high measurement noise: $\sigma_w^2=10^{-1},\, \text{SNR}(\m{y})=3.5$~dB.}
\label{fig:example1_reconstruction_with_high_noise}
\end{figure}

%The final investigation for this example concerns the impact of the relaxation time $\tau$ on the estimation accuracy. For this, $\tau$ was varied between $10^{-12}\,s$ and $200 \cdot 10^{-12}\,s$. Fig.~\ref{fig:Example1_tau} shows the average RMSE (ARMSE) of $\hat{\ve{d}}$ and $\hat{\bm{\mu}}$ (averaged over all grid points) plotted over $\tau$. It can be observed that the estimation accuracy decreases for increasing values of $\tau$. This behavior is expected since an increasing $\tau$ results in less energy arriving at the sensor. 

% -------------------------------------------------------------------
% Demonstration 3
% -------------------------------------------------------------------
\subsection{Example for estimating a piecewise constant absorption profile}
\label{sec:Demonstration_3}

Instead of the smooth absorption profiles employed for the previous simulations, we now use a piecewise constant $\bm{\mu}$. Other changes compared to Tab.~\ref{tab:Example1} are listed in Tab.~\ref{tab:Example3}. 

\begin{table}[ht]
\begin{center}
	\caption{Model parameters used for the second demonstration example. \label{tab:Example3} }
  \begin{tabular}{ | c | c | c | }
  	\hline
  	\textbf{Parameter} & \textbf{Value}  & \textbf{Unit} \\ 
    \hline \hline
    $\Delta_t$ 		& $10^{-7}$ & $\mathrm{s}$ 		\\ \hline
    $\Delta_z$ 		& $3\cdot 10^{-3}$ & $\mathrm{m}$ 		\\ \hline
    $N_\ve{z}$ 		& $100$ 			&\\ \hline
    $N_\ve{y}$ 		& $2000$ 		&	\\ \hline
  \end{tabular}
\end{center}
\end{table}

The area of interest now spans a depth of $N_\ve{z}\Delta_z = 3\,\mathrm{cm}$, while the laser modulation signal $\ve{i}$ is a chirp signal depicted in Fig.~\ref{fig:Example3fff} along with the corresponding measurement vector $\ve{y}$. The true absorption profile $\bm{\mu}$ as well as its estimates are shown in Fig.~\ref{fig:Example3}. In this case, the variance of the measurement noise $\sigma_w^2$ is equal to $71.4$~dB, which is worse than in the previous example, however, the estimation of the absorption profile $\hat{\bm{\mu}}$ is still very close to the true $\bm{\mu}$, especially for small indexes $n$. 

Again, the RMSE values of the estimations $\hat{\ve{d}}$ and $\hat{\bm{\mu}}$ averaged over $100$ simulation runs are shown in Fig.~\ref{fig:Example3_sdt_dev}. Similarly to Example~\ref{sec:Demonstration_1}, the errors tend to increase with increasing depth and the error of $\hat{\ve{d}}$ is usually smaller or equal to the error of $\hat{\bm{\mu}}$ due to the successive evaluation of the elements of $\hat{\mu}_n$ in \eqref{eq:mu_recursion}.

We also calculated the ARMSE values for varying levels of measurement noise, which can be seen in Fig.~\ref{fig:example4_opt}. Again, for noisy measurements the BLUE has the highest ARMSE, while the non-negative Tikhonov regularization provides the best estimation. In Fig.~\ref{fig:example4_opt_reconstruction}, we display a worst-case scenario with $\text{SNR}(\m{y})=1.5$~dB. Even for such a noisy measurement data with very low SNR, see e.g.\ Fig.~\ref{fig:Example4_mu_est_2}, our method provides a meaningful estimation of the true absorption profile in Fig.~\ref{fig:Example4_mu_est_3}.   

\begin{table}
%\begin{table*}
\centering
	\caption{Estimation accuracy for different laser modulation signal.}
	\scalebox{0.7}{
\begin{tabular}{|c||c|c|c|c|}
 \multicolumn{5}{c}{} \vspace{-0.5em} \\ \hline
\multirow{1}{*}{\parbox{0.38\textwidth}{\vspace{4mm}\textbf{Laser mod. signal}, $\sigma_w^2=10^{-10}$}} & \multicolumn{4}{|c|}{\textbf{ARMSE}$(\hat{\bm{\mu}})$ \textbf{of the estimations}} \\ \cline{2-5}
\textbf{} & 	\bigstrut \textbf{BLUE} & \textbf{Tikhonov}  & \textbf{non-negative Tikhonov} & \textbf{Damped SVD}\\ \hhline{|=====|}
Short pulse & $3.6\cdot 10^{-3}$ & $1.8\cdot 10^{-1}$ & $7.2\cdot 10^{-3}$ & $5.3\cdot 10^{-1}$ \\ \hline
Chirp & $1.2\cdot 10^{0}$ & $7.7\cdot 10^{-1}$ & $2.7\cdot 10^{-1}$ & $1.1\cdot 10^{0}$ \\ \hline
\end{tabular}
}
\label{tab:Example3_ARMSE_for_pulses}
\vspace{-5mm}
\end{table}

Finally, in order to demonstrate the impact of the laser modulation signal $\m{i}$ on the accuracy of the estimation, we repeated the simulation by using the same short pulse as in Example~\ref{sec:Demonstration_1}. For a fair comparison, we set the energy of the short pulse to be equal to the energy of the chirp signal. It can be seen that the resulting ARMSE of the estimations $\hat{\bm{\mu}}$ in Tab.~\ref{tab:Example3_ARMSE_for_pulses} are better for the short pulse in this example. In the following section we reshape the laser modulation signal $\m{i}$ in such a way that the error of the estimation is minimized. %In the next example, we will demonstrate this procedure by utilizing the optimization method in Section~\ref{sec:Optimization}.

%In order to obtain a fair comparison, we set the duration of the short pulse in such a way that it matches the duration of the oscillation of the chirp at the highest frequency. Furthermore, the energy of the short pulse signal and the chirp pulse are equal. The ARMSE of the estimations $\hat{\bm{\mu}}$ can be found in Tab.\~ref{tab:Example3_ARMSE_for_pulses}. For these modified settings and modified laser modulation signal we obtain the values $3.1$ and $3.7$ for the ARMSE of $\hat{\ve{d}}$ and $\hat{\bm{\mu}}$, respectively. This shows that in this example, the short pulse obtains slightly better performance than the chirp when using the proposed reconstruction method. 
%In the next section, we seek for an optimum laser modulation signal $\ve{i}$ that minimizes the ARMSE.

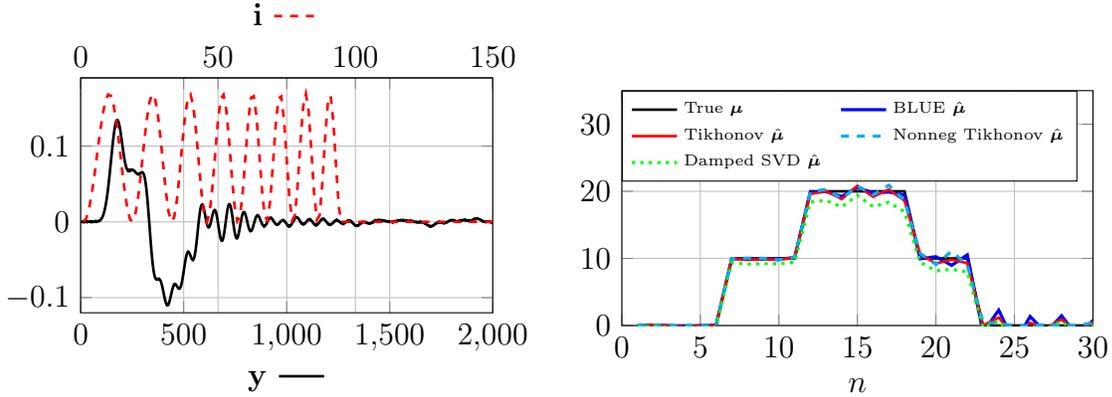
\begin{figure}[!t]
\centering
\vspace{2mm}
  \subfigure[The laser modulation signal $\ve{i}$ and the measured pressure signal $\ve{y}$.]{
	\label{fig:Example3fff}
\begin{tikzpicture}
\pgfplotsset{set layers}
\begin{axis}[
width=0.45\columnwidth, height = .3\columnwidth,
xmin=0,xmax=2000,
ymin=-0.12,ymax=0.19,
xlabel style = {align=center},
xlabel={$\ve{y}$ \ref{p1} },
%unit vector ratio*=1 200 1,
        legend style={at={(1,0)}, anchor=south east,
								/tikz/column 2/.style={
                column sep=5pt,
            },
        font=\tiny},
grid=major,
]
%\addplot[line width=1pt, color=black] table[x index =0, y index =1] {./Single_estimation_result_y_ex1.dat};
\addplot[line width=1pt, color=black] table[x index =0, y index =1] {./Example_4_y_and_low_noise.dat};
%\addlegendentry{{Measurement $\m{y}$}}
%\addplot[line width=1pt, color=red, style=dashed] table[x index =0, y index =3] {./Example_4_y_and_low_noise.dat};
%\addlegendentry{{Laser mod. sig. $\m{i}$}}
\label{p1}
\end{axis}
\begin{axis}[
width=0.45\columnwidth, height = .3\columnwidth,
xmin=0,xmax=150,
ymin=-0.12,ymax=0.19,
axis y line=none,
axis x line*=top,
xlabel style = {align=center},
xlabel={$\ve{i}$ \ref{fig:Example3_chirp} },
%unit vector ratio*=4/3 1 1,
        legend style={at={(1,0)}, anchor=south east,
								/tikz/column 2/.style={
                column sep=5pt,
            },
        font=\tiny},
grid=major,
]
\addplot[line width=1pt, color=red, style=dashed] table[x index =0, y index =3] {./Example_4_y_and_low_noise.dat};
\label{fig:Example3_chirp}
\end{axis}
\end{tikzpicture}%  \label{fig:Example1}
  } \hspace{0mm}
	\subfigure[The true and estimated absorption profile.]{
	\label{fig:Example3}	
	\begin{tikzpicture}
    \begin{axis}[%
    compat=newest, 
		width=0.5\columnwidth, height = .3\columnwidth,
		ylabel style={align=center}, 
		ylabel style={text width=3.4cm}, 
		xlabel=$n$,
		legend pos=north west, 
		legend cell align=left,
		legend columns=2, 
        legend style={at={(0,1)}, anchor=north west,
								/tikz/column 2/.style={
                column sep=5pt,
            },
        font=\tiny},
		xmin = 0,
		xmax = 30,
		ymax = 35,
		ymin = 0,
		grid=major]
			\addplot[line width=1pt, color=black] table[x index =0, y index =1] {./Example_4_low_noise_mu.dat};
			\addlegendentry{True $\bm{\mu}$}
			\addplot[line width=1.2pt, color=blue] table[x index =0, y index =2] {./Example_4_low_noise_mu.dat};
			\addlegendentry{BLUE $\hat{\bm{\mu}}$}
			\addplot[line width=1pt, color=red] table[x index =0, y index =3] {./Example_4_low_noise_mu.dat};
			\addlegendentry{Tikhonov $\hat{\bm{\mu}}$}
			\addplot[line width=1.2pt, color=cyan, style=dashed] table[x index =0, y index =4] {./Example_4_low_noise_mu.dat};
			\addlegendentry{Nonneg Tikhonov $\hat{\bm{\mu}}$}
			\addplot[line width=1.2pt, color=green, style=dotted] table[x index =0, y index =5] {./Example_4_low_noise_mu.dat};
			\addlegendentry{Damped SVD $\hat{\bm{\mu}}$}		
    \end{axis}	
	\end{tikzpicture}
  } \hspace{0mm}
	\centering
\caption{Application of the linear model for estimating the absorption profile.}
\label{fig:example3_reconstruction}
\end{figure}

\begin{figure}[!t]
\centering
\vspace{2mm}
  \subfigure[Relation between the RMSE of $\hat{\ve{d}}$ and $\hat{\bm{\mu}}$ for the BLUE and the Tikhonov regularization.]{
\begin{tikzpicture}

    \begin{axis}[%
		  %scaled y ticks=base 10:-2,
    	name=plot1,
    	compat=newest, 
		width=0.5\columnwidth, height = .3\columnwidth, xlabel=$n$, 
		ylabel style={align=center}, 
		ylabel style={text width=3.4cm},
		%ylabel={$y_k$}, 
		%scaled y ticks=base 10:-3,
		legend pos=north west, 
		legend cell align=left,
		legend columns=1, 
        legend style={at={(0,1)}, anchor=north west,
                    % the /tikz/ prefix is necessary here...
                    % otherwise, it might end-up with `/pgfplots/column 2`
                    % which is not what we want. compare pgfmanual.pdf
            /tikz/column 2/.style={
                column sep=5pt,
            },
        font=\tiny},
		xmin = 0,
		xmax = 30,
		ymax = 4.5e0,
		ymin = 0,
		grid=major]
			\addplot[line width=1pt, color=blue, style=dotted] table[x index =0, y index =1] {./Example_4_ARMSE_low_noise.dat};
			\addlegendentry{{BLUE $\text{RMSE}(\hat{\bm{\mu}})$}}
			
			\addplot[line width=0.5pt, color=blue] table[x index =0, y index =2] {./Example_4_ARMSE_low_noise.dat};
			\addlegendentry{{BLUE $\text{RMSE}(\hat{\ve{d}})$}}
			
			\addplot[line width=1pt, color=red, style=dotted] table[x index =0, y index =3] {./Example_4_ARMSE_low_noise.dat};
			\addlegendentry{{Tikhonov $\text{RMSE}(\hat{\bm{\mu}})$}}
			
			\addplot[line width=0.5pt, color=red] table[x index =0, y index =4] {./Example_4_ARMSE_low_noise.dat};
			\addlegendentry{{Tikhonov $\text{RMSE}(\hat{\ve{d}})$}}
						
    \end{axis}
    \label{fig:Example3_sdt_dev_a}
\end{tikzpicture}
  } \hspace{0mm}
	\subfigure[Relation between the RMSE of $\hat{\ve{d}}$ and $\hat{\bm{\mu}}$ for the non-negative Tikhonov regularization and the damped SVD.]{
\begin{tikzpicture}

    \begin{axis}[%
		  %scaled y ticks=base 10:-2,
    	name=plot1,
    	compat=newest, 
		width=0.5\columnwidth, height = .3\columnwidth, xlabel=$n$, 
		ylabel style={align=center}, 
		ylabel style={text width=3.4cm},
		%ylabel={$y_k$}, 
		%scaled y ticks=base 10:-3,
		legend pos=north west, 
		legend cell align=left,
		legend columns=1, 
        legend style={at={(0,1)}, anchor=north west,
                    % the /tikz/ prefix is necessary here...
                    % otherwise, it might end-up with `/pgfplots/column 2`
                    % which is not what we want. compare pgfmanual.pdf
            /tikz/column 2/.style={
                column sep=5pt,
            },
        font=\tiny},
		xmin = 0,
		xmax = 30,
		ymax = 4.5e0,
		ymin = 0,
		grid=major]
			\addplot[line width=1pt, color=cyan, style=dotted] table[x index =0, y index =5] {./Example_4_ARMSE_low_noise.dat};
			\addlegendentry{{Nonneg Tikhonov $\text{RMSE}(\hat{\bm{\mu}})$}}
			
			\addplot[line width=1pt, color=cyan] table[x index =0, y index =6] {./Example_4_ARMSE_low_noise.dat};
			\addlegendentry{{Nonneg Tikhonov $\text{RMSE}(\hat{\ve{d}})$}}
			
			\addplot[line width=1pt, color=green, style=dotted] table[x index =0, y index =7] {./Example_4_ARMSE_low_noise.dat};
			\addlegendentry{{Damped SVD $\text{RMSE}(\hat{\bm{\mu}})$}}
			
			\addplot[line width=0.5pt, color=green] table[x index =0, y index =8] {./Example_4_ARMSE_low_noise.dat};
			\addlegendentry{{Damped SVD $\text{RMSE}(\hat{\ve{d}})$}}
						
    \end{axis}
    \label{fig:Example3_sdt_dev_b}
\end{tikzpicture}
  } \hspace{0mm}
	\centering
\caption{RMSE of $\hat{\ve{d}}$ and $\hat{\bm{\mu}}$  averaged over $100$ simulation runs for all grid points within the area of interest $n = 0, \hdots, N_\ve{z}-1$ with average $\text{SNR}(\m{y})=71.4$.}
\label{fig:Example3_sdt_dev}
\end{figure}
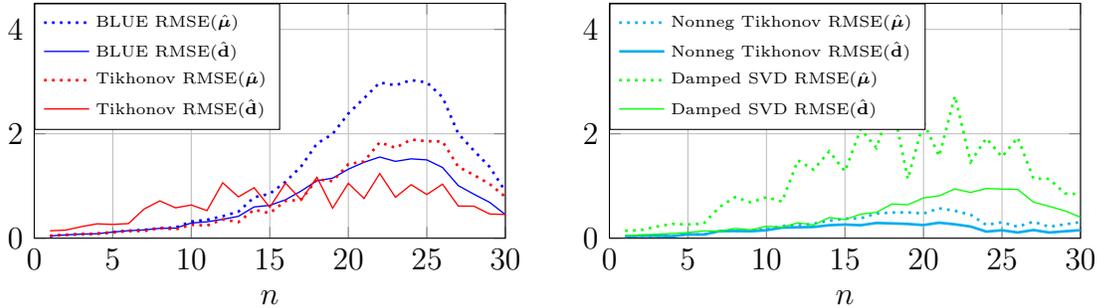

\begin{figure}[!t]
\centering
\subfigure[ARMSE of $\hat{\m{d}}$.]{
	\begin{tikzpicture}
    \begin{loglogaxis}[%
   	name=plot1,
   	compat=newest, 
		unit vector ratio*=0.4 1 1,
		scale=0.8,
		xlabel=$\sigma_w^2$, 
		ylabel style={align=center}, 
		ylabel style={text width=3.4cm},
		scaled x ticks=base 10:-1,
%		scaled y ticks=base 10:2,
		legend pos=north west, 
		legend cell align=left,
		legend columns=1, 
		xmin=1e-10, xmax=1e-1,
	  ymin=1e-1, ymax=1.9e+2,
        legend style={at={(0,1)}, anchor=north west,
                    % the /tikz/ prefix is necessary here...
                    % otherwise, it might end-up with `/pgfplots/column 2`
                    % which is not what we want. compare pgfmanual.pdf
            /tikz/column 2/.style={
                column sep=5pt,
            },
        font=\tiny},
		grid=major]
			\addplot[line width=1pt, color=blue] table[x index =0, y index =1] {./Example_4_ARMSE_mu.dat};
			\addlegendentry{{BLUE}}
			
			\addplot[line width=1pt, color=red] table[x index =0, y index =2] {./Example_4_ARMSE_mu.dat};
			\addlegendentry{{Tikhonov}}

			\addplot[line width=1pt, color=cyan, style=dashed] table[x index =0, y index =3] {./Example_4_ARMSE_mu.dat};
			\addlegendentry{{Nonneg Tikhonov}}

			\addplot[line width=1pt, color=green, style=dotted] table[x index =0, y index =4] {./Example_4_ARMSE_mu.dat};
			\addlegendentry{{Damped SVD}}

    \end{loglogaxis}
    \begin{loglogaxis}[%
   	name=plot1,
   	compat=newest, 
		unit vector ratio*=0.4 1 1,
		scale=0.8,
		axis y line=none,
		axis x line*=top,
		xlabel style={align=center}, 
		xlabel style={text width=3.4cm},
		xlabel={$\text{SNR}(\m{y})$ in dB}, 
		xtick=data,
		xticklabels={71.4, 61.4, 41.5, 21.5, 1.5, -18.5},
%		scaled x ticks=base 10:-1,
		legend pos=north west, 
		legend cell align=left,
		legend columns=1, 
		xmin=1e-10, xmax=1e-1,
	  ymin=1e-1, ymax=1.9e+2,
		]
    \end{loglogaxis}	
\end{tikzpicture}
\label{fig:Example4_mu_armse}
}
\subfigure[ARMSE of $\hat{\bm{\mu}}$.]{
	\begin{tikzpicture}
    \begin{loglogaxis}[%
   	name=plot1,
   	compat=newest, 
		unit vector ratio*=0.4 1 1,
		scale=0.8,
		xlabel=$\sigma_w^2$, 
		ylabel style={align=center}, 
		ylabel style={text width=3.4cm},
		scaled x ticks=base 10:-1,
%		scaled y ticks=base 10:2,
		legend pos=north west, 
		legend cell align=left,
		legend columns=1, 
		xmin=1e-10, xmax=1e-1,
	  ymin=1e-1, ymax=1.9e+2,
        legend style={at={(0,1)}, anchor=north west,
                    % the /tikz/ prefix is necessary here...
                    % otherwise, it might end-up with `/pgfplots/column 2`
                    % which is not what we want. compare pgfmanual.pdf
            /tikz/column 2/.style={
                column sep=5pt,
            },
        font=\tiny},
		grid=major]
			\addplot[line width=1pt, color=blue] table[x index =0, y index =1] {./Example_4_ARMSE_d.dat};
			\addlegendentry{{BLUE}}
			
			\addplot[line width=1pt, color=red] table[x index =0, y index =2] {./Example_4_ARMSE_d.dat};
			\addlegendentry{{Tikhonov}}

			\addplot[line width=1pt, color=cyan, style=dashed] table[x index =0, y index =3] {./Example_4_ARMSE_d.dat};
			\addlegendentry{{Nonneg Tikhonov}}

			\addplot[line width=1pt, color=green, style=dotted] table[x index =0, y index =4] {./Example_4_ARMSE_d.dat};
			\addlegendentry{{Damped SVD}}

    \end{loglogaxis}
    \begin{loglogaxis}[%
   	name=plot1,
   	compat=newest, 
		unit vector ratio*=0.4 1 1,
		scale=0.8,
		axis y line=none,
		axis x line*=top,
		xlabel style={align=center}, 
		xlabel style={text width=3.4cm},
		xlabel={$\text{SNR}(\m{y})$ in dB}, 
		xtick=data,
		xticklabels={71.4, 61.4, 41.5, 21.5, 1.5, -18.5},
		scaled x ticks=base 10:-1,
		legend pos=north west, 
		legend cell align=left,
		legend columns=1, 
		xmin=1e-10, xmax=1e-1,
	  ymin=1e-1, ymax=1.9e+2,
		]
    \end{loglogaxis}	
\end{tikzpicture}
\label{fig:Example4_d_armse}
}
\caption{ARMSE of the estimated quantities $\m{d}$ and $\bm{\mu}$ for Example 2.}%
\label{fig:example4_opt}%
\end{figure}
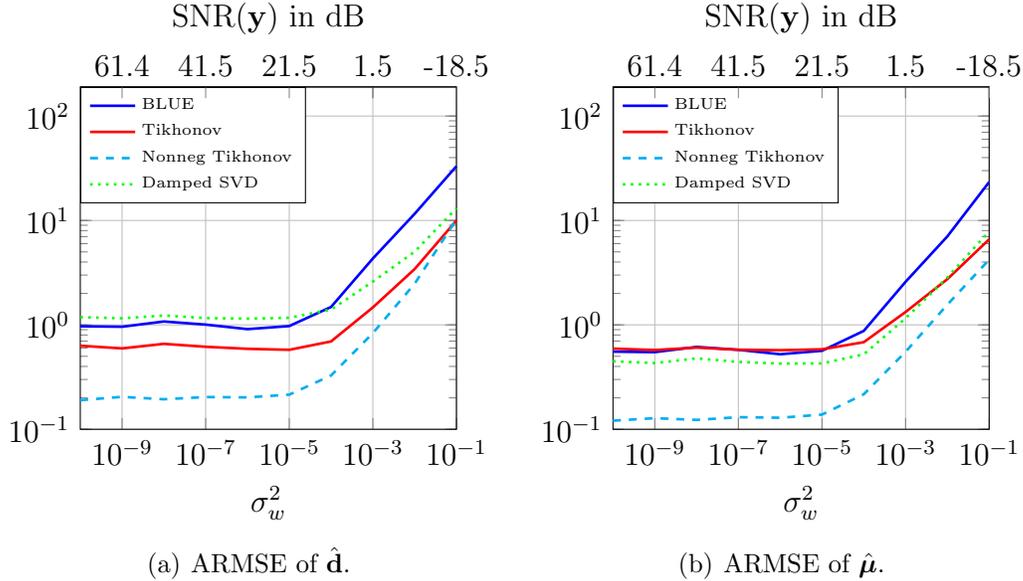

\begin{figure}[!t]
\centering
\vspace{4mm}
	\subfigure[The measurement vector with and without noise.]{
\begin{tikzpicture}
\pgfplotsset{set layers}
\begin{axis}[
width=0.45\columnwidth, height = .3\columnwidth,
xmin=0,xmax=2000,
ymin=-0.2,ymax=0.2,
xlabel=$k$,
%unit vector ratio*=1 20 1,
        legend style={at={(1,0)}, anchor=south east,
								/tikz/column 2/.style={
                column sep=5pt,
            },
        font=\tiny},
grid=major,
]
%\addplot[line width=1pt, color=black] table[x index =0, y index =1] {./Single_estimation_result_y_ex1.dat};
\addplot[line width=1pt, color=black] table[x index =0, y index =1] {./Example_4_y_and_noise.dat};
\addlegendentry{{Noisy measurement $\m{y}$}}
\addplot[line width=0.8pt, color=red] table[x index =0, y index =2] {./Example_4_y_and_noise.dat};
\addlegendentry{{Noise free measurement $\m{y}$}}
\end{axis}
\label{fig:Example4_mu_est_2}
\end{tikzpicture}
  } \hspace{0mm}
	\subfigure[The true and estimated absorption profiles.]{

\begin{tikzpicture}
    \begin{axis}[%
    compat=newest, 
		width=0.5\columnwidth, height = .3\columnwidth,
		ylabel style={align=center}, 
		ylabel style={text width=3.4cm}, 
		xlabel=$n$,
		legend pos=north west, 
		legend cell align=left,
		legend columns=2, 
        legend style={at={(0,1)}, anchor=north west,
								/tikz/column 2/.style={
                column sep=5pt,
            },
        font=\tiny},
		xmin = 0,
		xmax = 30,
		ymax = 32,
		ymin = -2,
		grid=major]
			\addplot[line width=1pt, color=black] table[x index =0, y index =1] {./Example_4_high_noise_mu.dat};
			\addlegendentry{True $\bm{\mu}$}
%			\addplot[line width=1.2pt, color=blue] table[x index =0, y index =2] {./Example_1_high_noise_mu.dat};
%			\addlegendentry{BLUE $\hat{\bm{\mu}}$}
			\addplot[line width=1pt, color=red] table[x index =0, y index =3] {./Example_4_high_noise_mu.dat};
			\addlegendentry{Tikhonov $\hat{\bm{\mu}}$}
			\addplot[line width=1.2pt, color=cyan, style=dashed] table[x index =0, y index =4] {./Example_4_high_noise_mu.dat};
			\addlegendentry{Nonneg Tikhonov $\hat{\bm{\mu}}$}
			\addplot[line width=1.2pt, color=green, style=dotted] table[x index =0, y index =5] {./Example_4_high_noise_mu.dat};
			\addlegendentry{Damped SVD $\hat{\bm{\mu}}$}		
    \end{axis}
		\label{fig:Example4_mu_est_3}	
	\end{tikzpicture}
  }
\caption{Application of the linear model for estimating the absorption profile with high measurement noise: $\sigma_w^2=10^{-3},\, \text{SNR}(\m{y})=1.5$ dB.}
\label{fig:example4_opt_reconstruction}
\end{figure}
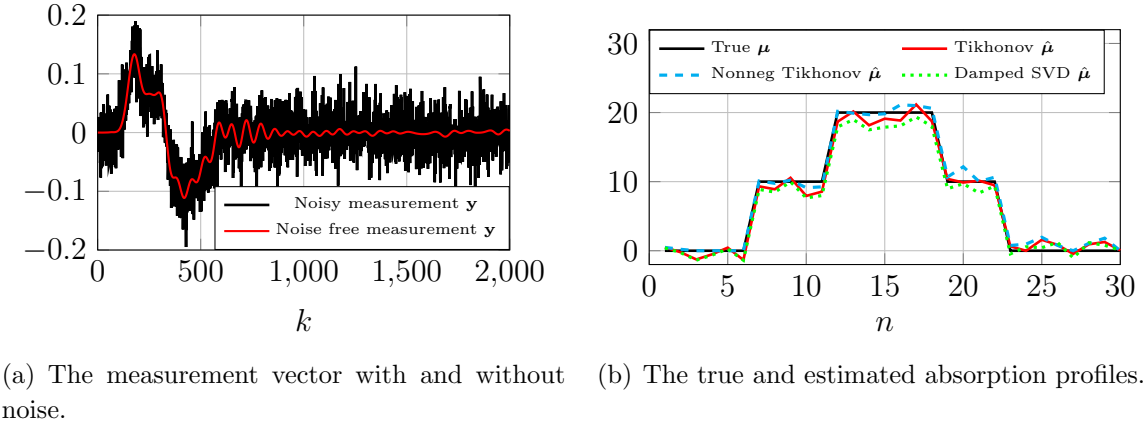

% -------------------------------------------------------------------
% Demonstration 4
% -------------------------------------------------------------------
\subsection{Example for optimizing the laser modulation signal}
\label{sec:Demonstration_4}

Now we demonstrate the performance gain using an optimized laser modulation signal obtained by the methods described in Section~\ref{sec:Optimization}. It will turn out that the resulting optimized laser modulation signal strongly depends on the considered problem. In order to demonstrate this relation, two different experiments with same optimization constraints are performed. More precisely, we utilize the following constraints: 
\begin{enumerate}
\item The elements of $\ve{i}$ must be larger than or equal to $0$ and smaller than or equal to $1$. 
\item $\ve{i}$ must be time-limited with a maximum length $N_\ve{i}=50$, or $N_\ve{i}=100$ for Examples~\ref{sec:Demonstration_1}-\ref{sec:Demonstration_3}, respectively. 
\item In order to provide a fair comparison between different laser modulation schemes, their energy shall be equal. This constraint is implemented via
\begin{align}
 \left\|\ve{i}\right\|_2^2=\sum_{k=0}^{N_{\ve{i}}-1} i_k^2 = 1. \label{equ:RECENDT_039}
\end{align} 
\item The last constraint is that $\ve{i}$ shall be band-limited. This band-limitation is implemented the following way. Let $\m{B}_1$ be a matrix that inserts $N_\text{zeros}$ number of zeros at the beginning and the end of $\ve{i}$
\begin{align}
\m{B}_1 = \begin{bmatrix} \m{0}^{N_{\text{zeros}} \times N_{\ve{i}}} \\ \m{I}^{N_{\ve{i}} \times N_{\ve{i}}} \\ \m{0}^{N_{\text{zeros}} \times N_{\ve{i}}} \end{bmatrix} \in \mathbb{R}^{N_{\ve{i}} + 2N_{\text{zeros}} \times N_{\ve{i}}}.
\end{align}
We utilize the discrete Fourier transform (DFT) to transform $\m{B}_1\ve{i}$ into the frequency domain. Let $\m{F}$ denote the DFT matrix of size $N_\text{dft} \times N_\text{dft}$ with $N_\text{dft} = N_{\ve{i}} + 2 N_\text{zeros}$. The product $\ve{q} = \m{F} \m{B}_1  \ve{i}$ produces a double-sided discrete spectrum of $\ve{i}$. The single-sided spectrum $\ve{q}_\text{s} \in \mathbb{C}^{N_\mathrm{s} \times 1}$ has a length $N_\text{s}$ 
\begin{align}
N_\mathrm{s} = \begin{cases}   N_\text{dft}/2+1  & N_\text{dft} \text{ is even} \\
 (N_\text{dft}+1)/2  & N_\text{dft}  \text{ is odd} \end{cases}\,, \label{equ:RECENDT_023q}
\end{align}
and $\ve{q}_\text{s}$ is given by 
\begin{align}
\ve{q}_\text{s} = \m{B}_2 \ve{q} = \m{B}_2 \m{F} \m{B}_1  \ve{i}\,, \label{equ:RECENDT_023w}
\end{align}
where $ \m{B}_2 = \begin{bmatrix} \m{B}_3 & \m{0}^{N_\mathrm{s} \times N_\text{dft} - N_\mathrm{s}}  \end{bmatrix}$ and where $\m{B}_3$ is a diagonal matrix of size $N_\mathrm{s} \times N_\mathrm{s}$ whose first diagonal element is $1$ and all remaining diagonal elements are $2$.  A possible way of generating a band-limited $\ve{i}$ is to enforce
\begin{align}
\text{abs}\left( \ve{q}_\mathrm{hf} \right) = \text{abs}\left(\m{B}_4 \ve{q}_\text{s} \right) \leq \bm{\varepsilon}, \label{equ:RECENDT_040} 
\end{align} 
where $\m{B}_4 = \begin{bmatrix} \m{0}^{N_\mathrm{hf} \times N_\mathrm{s} - N_\mathrm{hf}} & \m{I}^{N_\mathrm{hf} \times N_\mathrm{hf}}  \end{bmatrix}$ is a matrix that sorts out the $N_\mathrm{hf}$ highest frequencies of the single-sided spectrum, and $\bm{\varepsilon} \in \mathbb{R}^{N_\mathrm{hf} \times 1}$ is a vector with positive but arbitrary small values. In this example, we consider a uniform $\bm{\varepsilon}$ according to $\bm{\varepsilon}  =\epsilon \ve{1}$, where $\ve{1}$ is a column vector of length $N_\mathrm{hf}$ with all elements being $1$. 
\end{enumerate}

We mainly use the same settings as for Examples~\ref{sec:Demonstration_1}-\ref{sec:Demonstration_3} with additional parameters for the optimization process that are listed in Tab.~\ref{tab:Example4}. For Example~\ref{sec:Demonstration_1}, the optimization process was initialized with a random sequence of length $N_\ve{i}=50$ and with values between $0$ and $1$. In Fig.~\ref{fig:Example4}, the optimal laser modulation signal $\ve{i}_\text{opt}$ is indicated by the red curve, while Fig.~\ref{fig:Example4_spectra} shows the constrained single-sided spectra derived according to \eqref{equ:RECENDT_023w}. Inspecting $\ve{i}_\text{opt}$ for Example~\ref{sec:Demonstration_1} reveals that, within the scope of this simulation and the utilized constraints, it is better to use a dense sequence of short pulses than a single pulse as laser modulation signal. The performance gain of the optimized laser modulation signal is also remarkable. Tab.~\ref{tab:Example4_ARMSE} presents the ARMSE over 100 simulation runs for the short pulse, the chirp, and the optimized laser modulation signal $\ve{i}_\text{opt}$. This clearly demonstrates that the optimized laser modulation signal is superior to the others in terms of the reconstruction accuracy. For Example~\ref{sec:Demonstration_1}, the table reveals that the optimization process reduces the ARMSE values by approximately one order of magnitude compared to the short pulse and two order of magnitudes compared to the chirp signal. For Example~\ref{sec:Demonstration_3}, the short pulse and the chirp laser modulation provided good estimations (cf.\ Tab.~\ref{tab:Example3_ARMSE_for_pulses}), which was not improved by the optimization in case of low noise level, i.e.\ $\sigma_w^2=10^{-10}$. To this end, we considered the worst-case scenario with high measurement noise, i.e.\ $\sigma_w^2=10^{-3},\,\text{SNR}(\m{y})=3.5$~dB. Now the optimization process was initialized by the chirp signal. Tab.~\ref{tab:Example4_ARMSE} shows again an improvement in the estimation accuracy for $\ve{i}_\text{opt}$, which is indicated by the black curve in Fig.~\ref{fig:Example4}.

%The short pulse is displayed in Fig.~\ref{fig:Example4} together with the result of the optimization process $\ve{i}_\text{opt}$. $N_\text{hf}$ and $\epsilon$ in Tab.~\ref{tab:Example4} were chosen such that the maximum slopes of the short pulse and $\ve{i}_\text{opt}$ are approximately equal. The reason for that choice is that it makes the maximum values of $u_k$, which serve as input signal for the SSM in \eqref{equ:RECENDT_002}, approximately equal. Allowing higher slopes of $i_k$ would result in larger values $u_k$ and in an increased signal to noise ratio of the measurements $\ve{y}$ in \eqref{equ:RECENDT_003}. Inspecting $\ve{i}_\text{opt}$ in Fig.~\ref{fig:Example4} reveals that, within the scope of this example and the utilized constraints, it is better to use a dense sequence of short pulses than a single pulse as laser modulation signal. The single-sided spectra of the short pulse and the optimized laser modulation signal are derived according to \eqref{equ:RECENDT_023w} and are shown in Fig.~\ref{fig:Example4_spectra}. This figure shows that the constraint in \eqref{equ:RECENDT_040} is fulfilled. 

\begin{table}[h]
\begin{center}
	\caption{Model parameters used for the third example. \label{tab:Example4} }
  \begin{tabular}{ | c | c |  }
  	\hline
  	\textbf{Parameter} & \textbf{Value}   \\
    \hline \hline
    $N_\ve{i}$ 		& $50\,(\text{or } 100)$  		 	\\ \hline
    $N_\text{hf}$ 	& $15$ 				\\ \hline
    $N_\text{zeros}$ 	& $5$ 			\\ \hline
    $\epsilon$ 		& $10^{-3}$ \\ \hline
  \end{tabular}
\end{center}
\end{table}

\begin{figure}[!t]
\centering
\vspace{2mm}
  \subfigure[Results of the optimization process.]{
\label{fig:Example4} 
\begin{tikzpicture}
    \begin{axis}[%
    	name=plot1,
    	compat=newest, 
		width=0.43\columnwidth, height = .3\columnwidth, xlabel=$k$, 
		ylabel style={align=center}, 
		ylabel style={text width=3.4cm},
		legend pos=north east, 
		legend cell align=left,
		legend columns=1, 
    legend style={at={(1,1)}, anchor=north east,
                    % the /tikz/ prefix is necessary here...
                    % otherwise, it might end-up with `/pgfplots/column 
                    % which is not what we want. compare pgfmanual.pdf
            /tikz/column 2/.style={
                column sep=5pt,
            },
        font=\small},
		xmin = 0,
		xmax = 100,
		ymax = 0.5,
		ymin = 0,
		grid=major]
			\addplot[line width=1pt, color=red, style=dashed] table[x index =0, y index =1] {./Opt_Laser_Mod_Sigs_1.dat};
			\addlegendentry{{$\ve{i}_\text{opt}$ for Example~\ref{sec:Demonstration_1}}}
			\addplot[line width=1pt, color=black] table[x index =0, y index =1] {./Opt_Laser_Mod_Sigs_4.dat};
			\addlegendentry{{$\ve{i}_\text{opt}$ for Example~\ref{sec:Demonstration_3}}}
    \end{axis}
\end{tikzpicture}
  } \hspace{0mm}
	\subfigure[Constrained single-sided spectra of the optimized laser modulation signals.]{
	\label{fig:Example4_spectra}	
\begin{tikzpicture}
    \begin{axis}[%
    	name=plot1,
    	compat=newest, 
		width=0.43\columnwidth, height = .3\columnwidth, xlabel=Frequency (GHz), 
		ylabel style={align=center}, 
		ylabel style={text width=3.4cm},
		ylabel={$|\ve{q}_\text{s}|$}, 
		legend pos=north east, 
		legend cell align=left,
		legend columns=1, 
    legend style={at={(1,1)}, anchor=north east,
                    % the /tikz/ prefix is necessary here...
                    % otherwise, it might end-up with `/pgfplots/column 
                    % which is not what we want. compare pgfmanual.pdf
            /tikz/column 2/.style={
                column sep=5pt,
            },
        font=\small},
		xmin = 0,
		xmax = 3e-1,
		ymax = 9,
		ymin = 0,
		grid=major]
			\addplot[line width=1pt, color=red, style=dashed] table[x index =1, y index =2] {./Opt_Laser_Mod_Sigs_spectra_1.dat};
			\addlegendentry{{$\ve{i}_\text{opt}$ Spectra~\ref{sec:Demonstration_1}}}
			
			\addplot[line width=1pt, color=black] table[x index =1, y index =2] {./Opt_Laser_Mod_Sigs_spectra_4.dat};
			\addlegendentry{{$\ve{i}_\text{opt}$ Spectra~\ref{sec:Demonstration_3}}}
    \end{axis}
\end{tikzpicture}
%----------------For LogPlot------------------
%\begin{tikzpicture}
    %\begin{semilogyaxis}[%
    	%name=plot1,
    	%compat=newest, 
		%width=0.43\columnwidth, height = .3\columnwidth, xlabel=Frequency index, 
		%ylabel style={align=center}, 
		%ylabel style={text width=3.4cm},
		%ylabel={$|\ve{q}_\text{s}|$}, 
		%legend pos=north east, 
		%legend cell align=left,
		%legend columns=1, 
    %legend style={at={(1,1)}, anchor=north east,
                    %% the /tikz/ prefix is necessary here...
                    %% otherwise, it might end-up with `/pgfplots/column 
                    %% which is not what we want. compare pgfmanual.pdf
            %/tikz/column 2/.style={
                %column sep=5pt,
            %},
        %font=\small},
		%xmin = 0.8e-3,
		%xmax = 60,
		%ymax = 1.7e1,
		%ymin = 1e-3,
		%ytick={1e-3, 1e-1, 1e1},
		%grid=major]
			%\addplot[line width=1pt, color=red, style=dashed] table[x index =0, y index =1] {./Opt_Laser_Mod_Sigs_spectra_1_logscale.dat};
			%\addlegendentry{{$\ve{i}_\text{opt}$ Spectra~\ref{sec:Demonstration_1}}}
			%
			%\addplot[line width=1pt, color=black] table[x index =0, y index =1] {./Opt_Laser_Mod_Sigs_spectra_4_logscale.dat};
			%\addlegendentry{{$\ve{i}_\text{opt}$ Spectra~\ref{sec:Demonstration_3}}}
    %\end{semilogyaxis}
%\end{tikzpicture}
  } \hspace{0mm}
	\centering
\caption{Properties of the optimized laser modulation signals.}
\label{fig:example4_opt_laser_mod_sig}
\end{figure}
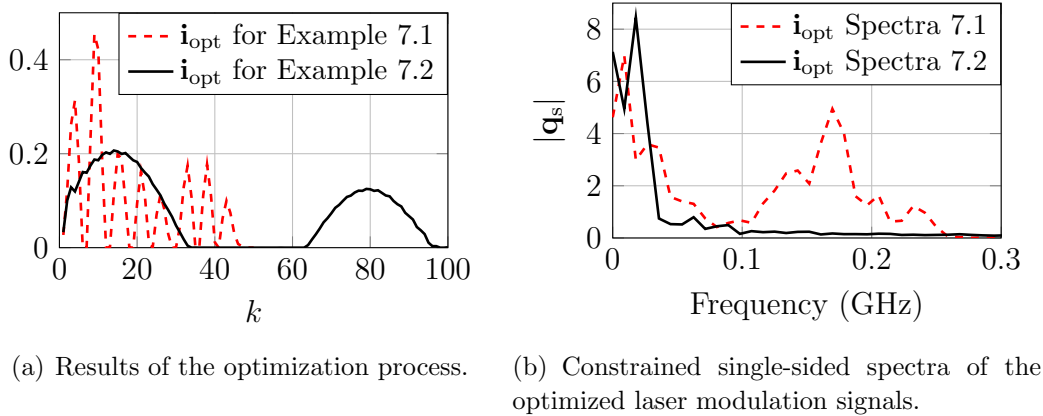

\begin{table}
\centering
	\caption{Estimation accuracy for different laser modulation signal.}
	\scalebox{0.7}{
\begin{tabular}{|c||c|c||c|c|}
 \multicolumn{5}{c}{} \vspace{-0.5em} \\ \hline
\multirow{1}{*}{\parbox{0.30\textwidth}{\vspace{2mm}\textbf{Laser mod.\ signal} \\ \textbf{Example}~\ref{sec:Demonstration_1}, $\sigma_w^2=10^{-10}$}} & \multicolumn{2}{|c||}{\textbf{\textbf{$\textbf{ARMSE}(\hat{\m{d}})$ of the estimations}}} & \multicolumn{2}{|c|}{\textbf{\textbf{$\textbf{ARMSE}(\hat{\bm{\mu}})$ of the estimations}} \bigstrut} \\ \cline{2-5}
\textbf{} & 	\bigstrut \textbf{BLUE} & \textbf{non-negative Tikhonov}  & \textbf{BLUE} & \textbf{non-negative Tikhonov}\\ \hhline{|=====|}
Short pulse & $2.6\cdot 10^{1}$ & $1.3\cdot 10^{1}$ & $4.3\cdot 10^{1}$ & $2.2\cdot 10^{1}$ \\ \hline
Chirp & $8.8\cdot 10^{2}$ & $2.2\cdot 10^{2}$ & $1.5\cdot 10^{3}$ & $3.5\cdot 10^{2}$ \\ \hline
Optimized & $4.8\cdot 10^{0}$ & $3.2\cdot 10^{0}$ & $8.2\cdot 10^{0}$ & $5.3\cdot 10^{0}$ \\ \hline 
\multicolumn{5}{c}{} \vspace{-0.5em} \\ \hline
\multirow{1}{*}{\parbox{0.30\textwidth}{\vspace{2mm}\textbf{Laser mod.\ signal} \\ \textbf{Example}~\ref{sec:Demonstration_3}, $\sigma_w^2=10^{-3}$}} & \multicolumn{2}{|c||}{\textbf{\textbf{$\textbf{ARMSE}(\hat{\m{d}})$ of the estimations}}} & \multicolumn{2}{|c|}{\textbf{\textbf{$\textbf{ARMSE}(\hat{\bm{\mu}})$ of the estimations}}\bigstrut} \\ \cline{2-5}
\textbf{} & 	\bigstrut \textbf{BLUE} & \textbf{non-negative Tikhonov}  & \textbf{BLUE} & \textbf{non-negative Tikhonov}\\ \hhline{|=====|}
Short pulse & $2.1\cdot 10^{0}$ & $1.1\cdot 10^{0}$ & $3.8\cdot 10^{0}$ & $1.7\cdot 10^{0}$ \\ \hline
Chirp & $3.1\cdot 10^{0}$ & $7.5\cdot 10^{-1}$ & $5.1\cdot 10^{0}$ & $1.1\cdot 10^{0}$ \\ \hline
Optimized & $1.2\cdot 10^{0}$ & $4.6\cdot 10^{-1}$ & $2.1\cdot 10^{0}$ & $7.3\cdot 10^{-1}$ \\ \hline
\end{tabular}
}
\label{tab:Example4_ARMSE}
\vspace{-5mm}
\end{table}

\section{Conclusion}
We developed a method for reconstructing the absorption profile in photoacoustic imaging based on surface measurements of the ultrasound waves. For approximating the original Stokes' PDE we introduced a discrete linear SSM. This approximation accounts for frequency dependent attenuation of the ultrasound waves as well as a decrease in laser intensity due to absorption. Then we proved that the parameters of the SSM can be chosen in such a way that the model is asymptotically stable, observable and controllable. The conditions of these properties are simple, and thus they can be easily verified. We also emphasize that our algorithm is of general nature, namely it allows for inhomogeneous probes with an arbitrary absorption profile. In addition, the SSM allows to linearly estimate a certain vector that depends on the absorption profile. Based on this algorithm, the absorption profile is estimated via a non-linear routine. We provided several simulations that demonstrate the reconstruction accuracy of the proposed approach for different noise levels. In these experiments, various regularization methods were studied to overcome the ill-posedness of the problem. 

Furthermore, a method to optimize the laser modulation signal has been introduced such that the accuracy of the estimated absorption profile is maximized. Utilizing the optimized laser modulation signals may yield a significant increase in reconstruction accuracy compared to short pulses as well as chirp modulation. The concrete performance gain depends on the utilized constraints of the optimization process. For a limited frequency bandwidth the result of the reconstruction gets better if a single short pulse is separated into smaller pulses having in total the same energy as the single short pulse. Sometimes it might be advantageous to use several smaller pulses instead of one larger excitation pulse, e.g.\ to meet safety guidelines for the maximum light fluence or because of power limitation of the excitation laser. If adequate reconstruction is used, such as our proposed state space model, the reconstruction error is the same for more pulses with less amplitude if the total energy is the same. %If the smaller pulses are 'reshaped' such that the rate of rise and decay are the same as for the single pulse (bandwidth limitation), then the reconstruction even gets better compared to the single pulse excitation. The time lag between the smaller pulses is irrelevant to the quality of the reconstruction result. This is a consequence of the linearity of the equations describing our state space model. Chirped time intervals between the smaller short excitation pulses have no benefit for the reconstruction result compared to other time intervals if an adequate reconstruction algorithm such as our state space reconstruction is used.

Compensation of acoustic attenuation and dispersion in two or three dimensions can always be reduced to a one-dimensional problem in a two-stage process: first, for each detector location the ideal signal in the absence of attenuation is calculated from the measured signal. This one-dimensional reconstruction can be performed with our presented state space model approach. In a second step, any reconstruction method for photoacoustic tomography can be used for reconstructions in higher dimensions \cite{Riviere2005, LaRiviere2006}. In two- or three-dimensional photoacoustic imaging beside optical absorption also optical scattering reduces the light fluence with increasing depth. In a semi-infinite medium with constant illumination on its surface the decreasing illumination with depth can be described by an effective attenuation coefficient, including optical absorption and scattering \cite{Dean-Ben2011}. 

\ack
This work was supported by the “K-Project for non-destructive testing and tomography plus” supported by COMET program of FFG and the federal government of Upper Austria and Styria and by the project "multimodal and in-situ characterization of inhomogeneous materials" (MiCi) by the federal government of Upper Austria and the European Regional Development Fund (EFRE) in the framework of the EU-program IWB2020.

\section*{References}

\bibliographystyle{unsrt}
\bibliography{References}

\end{document}